\documentclass{aims}

\usepackage{txfonts}

\def\typeofarticle{Research article}
\def\currentvolume{22}
\def\currentissue{2}
\def\currentyear{2026}

\def\ppages{997--1033}
\def\DOI{10.3934/jimo.2026037}

\def\Received{27 September 2025}
\def\Revised{25 November 2025}
\def\Accepted{15 January 2026}
\def\Published{22 January 2026}

\numberwithin{equation}{section}

\usepackage{amsmath,amssymb,amsthm}
\usepackage{mathtools}
\usepackage{booktabs}
\usepackage{graphicx}
\usepackage{float}

\usepackage{algorithm}
\usepackage{algpseudocode}
\usepackage{adjustbox}

\usepackage{paralist}

\newcommand{\req}{\textbf{R}}
\newcommand{\aux}{\textbf{A}}

\newtheorem{theorem}{Theorem}[section]
\newtheorem{corollary}[theorem]{Corollary}

\newtheorem{lemma}[theorem]{Lemma}
\newtheorem{proposition}[theorem]{Proposition}

\theoremstyle{definition}
\newtheorem{definition}[theorem]{Definition}

\newtheorem{assumption}[theorem]{Assumption}

\begin{document}
\sloppy

\title{Mechanism design and equilibrium analysis of smart contract--mediated resource allocation}

\author{%
Jinho Cha\affil{1}\corrauth
Justin Yu\affil{2}
Eunchan Daniel Cha\affil{3}
Emily Haneul Yoo\affil{4}
Caedon Geoffrey\affil{1}
Hyoshin Song\affil{5}
}

\shortauthors{Cha et al.}

\address{%
\addr{\affilnum{1}}{Department of Computer Science, Gwinnett Technical College, Lawrenceville, GA, USA}
\addr{\affilnum{2}}{Scheller College of Business, Georgia Institute of Technology, Atlanta, GA, USA}
\addr{\affilnum{3}}{School of Biological Sciences, Georgia Institute of Technology, Atlanta, GA, USA}
\addr{\affilnum{4}}{North Gwinnett High School, Suwanee, GA, USA}
\addr{\affilnum{5}}{Oakton High School, Vienna, VA, USA}
}

\corraddr{Email: jcha@gwinnetttech.edu}

\begin{abstract}
Decentralized coordination and digital contracting are becoming essential in complex industrial systems, yet existing approaches often rely on ad-hoc heuristics or purely technical blockchain implementations without a rigorous economic foundation. This study developed a mechanism-design framework for smart contract--mediated resource allocation that jointly embeds efficiency, fairness, and resilience in decentralized coordination. We modeled agent interactions as a contract-clearing game under shared capacity constraints, established the existence and uniqueness of equilibrium, and proposed a decentralized price-adjustment algorithm with provable convergence suitable for real-time operation. Performance was evaluated through extensive synthetic simulations and validated using a representative real-world dataset. In addition to controlled experiments, a long-horizon empirical analysis using financial and macroeconomic data from 2006 to 2025 examined the mechanism under major economic regimes and shock conditions. Results showed that the proposed mechanism consistently reduces inequality and cost while maintaining near-optimal efficiency and rapid recovery following shocks, demonstrating dynamic stability beyond steady state.
\end{abstract}


\maketitle

\section{Introduction}

The rapid digitalization of industrial systems—captured broadly under the
umbrella of Industry~4.0—has accelerated the integration of cyber–physical
infrastructure with autonomous decision-making technologies. Applications
ranging from predictive maintenance and supply chain optimization to smart
grids and healthcare management increasingly require secure, transparent, and
real-time coordination across heterogeneous participants
\cite{kagermann2013,lu2017,ivanov2021,wu2024,zhang2024,paul2022}. However,
traditional contracting mechanisms—whether based on bilateral negotiation,
centralized intermediaries, or informal agreements—often suffer from
inefficiencies, delays, and susceptibility to opportunistic behavior
\cite{cachon2006,greenberg1987,beck2018}. These limitations motivate the
development of automated and verifiable protocols that can enforce resource
allocation and compliance without reliance on trusted third parties.

Across contemporary industrial sectors, inefficient or opaque allocation
mechanisms lead to measurable welfare losses. In decentralized \emph{energy
markets}, renewable generators and storage operators compete for limited
transmission capacity. In \emph{supply chain logistics}, firms bid for port
slots, warehouse space, or transport capacity. In \emph{healthcare}, regulators
must allocate scarce vaccines, ICU beds, and medical supplies under fairness
constraints. In each case, unverifiable coordination causes either overuse
(inefficiency) or exclusion (fairness loss). A smart-contract–mediated
mechanism can automate these interactions by embedding pricing and enforcement
logic directly into tamper-resistant code, transforming economic agreements
into auditable, self-executing processes.

Despite the growing deployment of blockchain-based smart contracts, most
existing implementations prioritize \emph{technical automation} over \emph{economic
optimality}. Many systems merely replicate legacy rules (e.g., first-come
first-served, proportional sharing) in a blockchain environment, thereby
preserving the inefficiencies and inequities inherent to the original process.
Without an explicit mechanism-design foundation, a smart contract can execute
transactions correctly yet still produce socially suboptimal or unstable
outcomes. This distinction—between automating a rule and \emph{designing} the rule
itself—is at the core of the mechanism-design gap addressed in this paper.

Smart contracts provide an opportunity to bridge this gap by encoding allocation
rules within verifiable, tamper-proof logic \cite{szabo1997,buterin2014,christidis2016,wang2019,xu2022}. 
Recent work demonstrates promising applications in supply chains \cite{ivanov2022}, 
energy markets \cite{farahani2020,mohsenian2023}, and public health \cite{who2021}, 
but the majority of studies emphasize technological feasibility, security, or 
engineering aspects. Far less attention has been devoted to their 
\emph{mechanism-design implications}—specifically, how to implement efficiency, 
fairness, and resilience guarantees in competitive, shock-prone environments.

From an analytical standpoint, game theory and mechanism design offer a natural
foundation. Prior research establishes equilibrium existence and efficiency in
resource allocation games \cite{myerson1981,gabay1980,moulin2003}, develops 
fairness–efficiency trade-offs \cite{bertsimas2011,kearns2019}, and analyzes 
regret bounds in dynamic or adversarial environments \cite{hazan2016,shalev2012}. 
Yet these contributions remain largely disconnected from the literature on 
blockchain-based coordination, which has focused primarily on distributed 
algorithms and consensus protocols \cite{tsitsiklis1986,bertsekas1997,ding2023}. 
To date, no unified framework has combined mechanism-design principles with 
verifiable blockchain execution to support robust, decentralized coordination 
under uncertainty.

\begin{figure}[!htpb]
\centering
\includegraphics[width=1.0\linewidth]{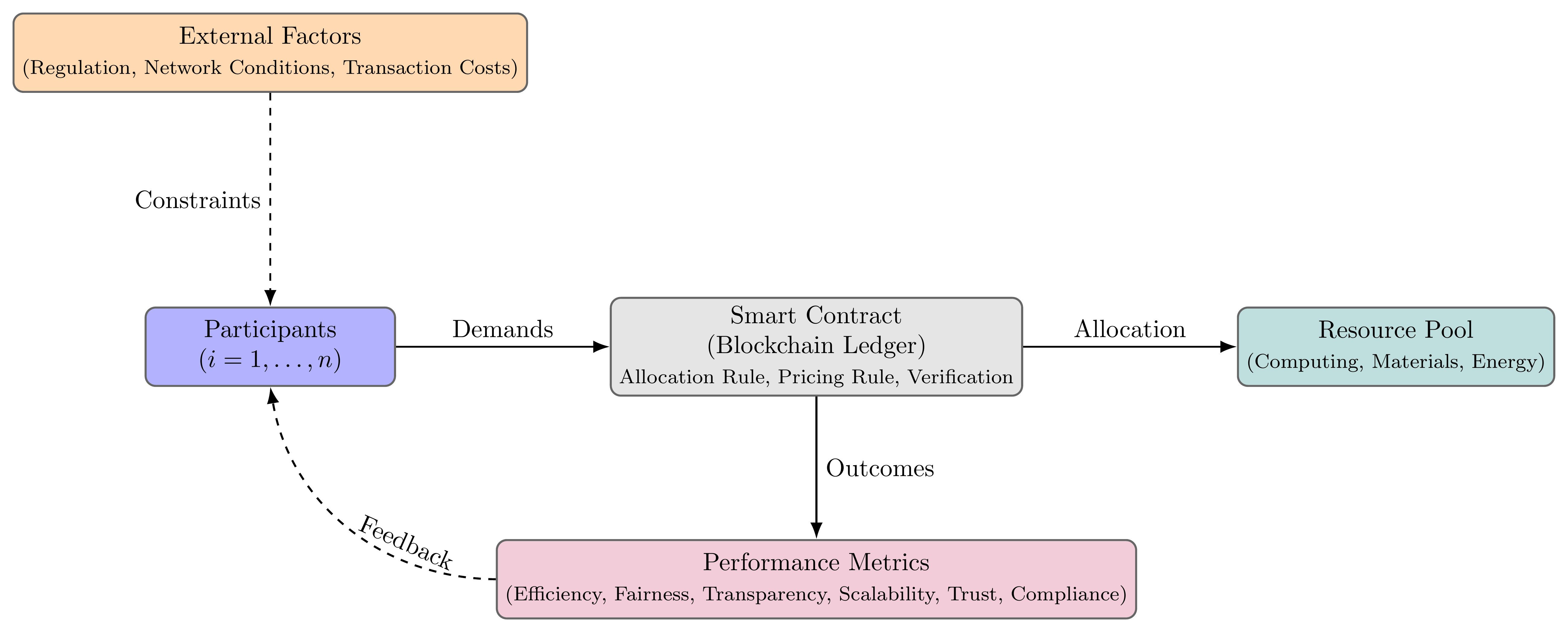}
\caption{Conceptual framework: participants submit demands to a blockchain-based
smart contract that allocates resources based on encoded rules, while external
shocks and feedback influence both demands and allocations. The resulting
outcomes—efficiency, fairness, and transparency—are fully auditable.}
\label{fig:framework}
\end{figure}

This paper fills this gap by developing a rigorous framework for
\emph{smart-contract–based mechanism design under shared capacity constraints}.
Participants submit demands to a contract that implements allocation and pricing
rules, subject to transaction and execution fees $(\tau, g)$. We formulate the
interaction as a non-cooperative game, derive the payoff structure, and
establish the following theoretical and managerial contributions:

\begin{itemize}
    \item \textbf{Existence and uniqueness of equilibrium:}
    Under mild convexity conditions, the contract-clearing game admits a unique,
    globally stable equilibrium implementable without central coordination.

    \item \textbf{Algorithmic convergence:}
    A decentralized price-adjustment algorithm is developed with provable
    convergence guarantees, enabling real-time operation in industrial systems.

    \item \textbf{Fairness–efficiency trade-off:}
    Pareto frontiers quantify efficiency gains versus equity, enabling regulators
    to tune contract parameters for sector-specific goals (e.g., energy vs.~logistics).

    \item \textbf{Shock–resilience guarantees:}
    Sublinear regret bounds demonstrate robustness to drift and policy shocks,
    showing that contract-mediated coordination remains stable under disruption.
\end{itemize}

Numerical simulations support these analytical results, revealing efficiency
gains up to 27\% and inequality reductions exceeding 40\% relative to
proportional allocation rules. Complementary sensitivity and robustness
analyses demonstrate that the mechanism behaves predictably under volatility,
participation shocks, and macroeconomic disturbances. Taken together, these
results show that smart contracts can function not merely as computational
tools but as \emph{institutional instruments} that promote transparency,
fairness, and resilience across major economic sectors—including energy,
logistics, healthcare, and public infrastructure
\cite{rai2019,halaburda2024,acharya2021}.

By aligning algorithmic coordination with mechanism-design principles, the
framework provides a deployable blueprint for regulators and platform designers
seeking equitable and efficient digital markets. It bridges the gap between
blockchain engineering and economic policy—enabling auditable resource-allocation 
systems that remain stable even under volatility, participation shocks, or 
policy changes.

\section{Literature on Smart Contracts and Mechanism Design}

The literature on resource allocation in digital and cyber--physical systems spans
mobile edge computing, cloud economics, blockchain-enabled coordination, and
information systems governance. While the technical foundations of these systems
are well established, their implications for trust, fairness, resilience, and
organizational legitimacy remain underexplored. We review three strands most
relevant to this study: (i) smart contracts in organizational systems,
(ii) game-theoretic approaches to resource allocation, and
(iii) mechanism design perspectives on efficiency and fairness.
Together, these literatures reveal a growing yet fragmented understanding of how
digital coordination systems can embed economic and ethical principles.

\subsection{Smart Contracts in Organizational Contexts}

Smart contracts have been widely studied as programmable agreements on blockchain
platforms, enabling tamper-resistant execution and transparent enforcement
\cite{szabo1997,buterin2014}. Early research emphasized cryptographic security,
consensus protocols, and distributed architectures \cite{cachin2016}, whereas
recent studies have extended these concepts to supply chain management,
manufacturing, and industrial automation \cite{casino2019,ivanov2021}.  
Within information systems research, blockchain has increasingly been theorized as a
\emph{governance mechanism} that reduces opportunism, enhances transparency, and
supports inter-organizational trust \cite{beck2018,pavlou2002,rai2019,halaburda2024}.
These contributions highlight that smart contracts are not merely computational
artifacts but institutional devices that redefine how rules and responsibilities are
enacted across organizational boundaries.  

From an industrial standpoint, adoption decisions depend not only on technical
performance but also on regulatory legitimacy and incentive compatibility.
Surveys of blockchain-enabled resource management in mobile and edge computing
\cite{rashid2025} demonstrate that firms adopt such systems only when algorithmic
allocation rules are perceived as fair, auditable, and efficient.  
However, most existing studies focus on the \emph{execution layer} of smart contracts—
how agreements are verified and enforced—rather than on their \emph{design layer},
which determines whether these digital rules achieve equilibrium and fairness.
This gap motivates a more formal game-theoretic analysis of contract behavior.

\subsection{Game-Theoretic Resource Allocation}

Game-theoretic models provide the analytical foundation for studying competition and
cooperation in resource allocation.  
In mobile edge and cloud computing, Muñoz et al.\ \cite{munoz2015} optimized radio and
computational resources under latency and energy constraints, while Dinh et al.\
\cite{dinh2013} examined multi-device offloading.
Zhang \cite{zhang2017} introduced stochastic games for dynamic task assignment, and
subsequent studies incorporated Stackelberg pricing \cite{li2019}, matching theory
for user--server association \cite{liu2018}, and reinforcement learning for adaptive
scheduling \cite{dai2020}.  
More recently, blockchain-based approaches integrate incentives with efficient
allocation in edge and vehicular networks
\cite{ding2023,zhang2024,yuan2024}.
Classic theoretical results remain highly influential:
Rosen \cite{rosen1965} established conditions for the existence and uniqueness of
concave $N$-person game equilibria, Gabay and Moulin \cite{gabay1980} analyzed
equilibrium stability, and Roughgarden and Tardos \cite{roughgarden2004} formalized
the price of anarchy in distributed systems.
In supply chains, Cachon and Netessine \cite{cachon2006} demonstrated that equilibrium
reasoning clarifies coordination failures, while Ivanov \cite{ivanov2022} emphasized
resilience under disruption shocks.

Collectively, these studies show that equilibrium concepts explain not only technical
efficiency but also institutional stability—how markets and organizations balance
competition and cooperation.  
Yet, most equilibrium models still assume either centralized price-setting or
exogenous coordination rules, offering limited insight into how decentralized digital
contracts can \emph{endogenously} achieve stability, equity, and convergence.
This limitation underscores the need to merge equilibrium theory with formal
mechanism-design principles.

\subsection{Mechanism Design and Fairness Considerations}

Mechanism design in operations research and economics seeks to construct allocation
rules that jointly satisfy efficiency, fairness, and incentive compatibility.
Generalized Nash equilibrium models have been applied to network pricing
\cite{kelly1997,cardellini2016}, and bilevel formulations have been developed for
profit-maximizing providers \cite{wang2019}.
The “price of fairness” concept in optimization quantifies efficiency losses incurred
when fairness constraints are imposed \cite{bertsimas2011}.  
In organizational sciences, fairness metrics such as the Gini index capture
distributional outcomes \cite{lambert2001}, while justice theory emphasizes that
perceptions of distributive and procedural fairness are essential for sustaining
trust and compliance \cite{greenberg1987}.
Recent information systems scholarship extends these insights to algorithmic
governance, emphasizing fairness, accountability, and legitimacy as cornerstones of
digital platforms \cite{rai2019,jobin2019}.
Parallel developments in computer science and data ethics examine algorithmic
fairness both as formal constraints and as human-centered perceptions
\cite{barocas-hardt-narayanan,mehrabi2021,kearns2019}.  

Despite these advances, integration of automated enforcement, incentive compatibility,
and distributive justice within a unified analytical framework remains limited.
Most existing mechanisms treat fairness as an ex-post evaluation rather than as an
ex-ante design principle.  
Bringing mechanism design into the domain of smart contracts enables these principles
to be embedded directly into digital allocation rules—linking the mathematical
structure of equilibria with normative goals of fairness and resilience.

\subsection{Research Gap and Positioning}

Synthesizing the above literatures reveals three critical gaps.  
First, while smart contracts promise automation and transparency, their function as
\emph{organizational allocation mechanisms} has rarely been analyzed within models
that simultaneously ensure efficiency, fairness, and incentive compatibility.  
Second, although equilibrium and optimization frameworks have advanced in mobile and
industrial systems, they generally omit governance, legitimacy, and resilience under
uncertainty—key requirements for sustainable digital coordination
\cite{ivanov2021,ivanov2022}.  
Third, there is no unified mechanism-design framework that integrates
game-theoretic equilibrium, fairness–efficiency trade-offs, and shock-resilience
analysis within one formal structure.

Addressing these gaps, this study develops a non-cooperative game of
smart-contract–mediated resource allocation, proves the existence and uniqueness of
contract equilibria, and introduces a decentralized price-adjustment algorithm with
provable convergence.  
By comparing the proposed contract mechanism with traditional proportional or
fixed-price rules, the study contributes to the literature on
operations research, mechanism design, and information systems governance.
More broadly, it provides a theoretical foundation for designing digital contracts
that act as \emph{institutional mechanisms}—embedding transparency, fairness, and
resilience directly into the architecture of decentralized industrial coordination.

\section{Contract Design for Efficient and Fair Industrial Resource Allocation}

Industrial systems such as supply chains, logistics networks, and production
platforms must allocate scarce resources across multiple agents. Traditional
allocation rules---whether proportional or centrally administered---often suffer
from inefficiency, opportunism, and lack of transparency. Digital contracts
implemented on blockchain platforms offer a compelling alternative: allocation
rules can be encoded, enforced automatically, and verified by all participants.
This section formalizes the contract design, introduces equilibrium concepts,
and develops performance metrics that jointly capture efficiency, fairness,
and resilience.

\subsection{Model Setup}

Let $N=\{1,\dots,n\}$ denote the set of industrial agents. Each agent $i$
requests a quantity $x_i \ge 0$, and we collect demands in the vector
$\mathbf{x}=(x_1,\dots,x_n)^\top$. The shared resource pool has capacity $m>0$:
\begin{equation}\label{eq:capacity}
    \mathbf{1}^\top \mathbf{x} \;\le\; m.
\end{equation}

\begin{assumption}[Valuation and Cost]\label{ass:val}
Each agent derives value $V_i(x_i)$ from consumption and incurs cost $C_i(x_i)$.
We impose:
\begin{enumerate}
    \item $V_i:\mathbb{R}_+\to\mathbb{R}$ is strictly concave, differentiable, 
    and satisfies $V_i'(0)=\infty$ (diminishing returns).
    \item $C_i:\mathbb{R}_+\to\mathbb{R}$ is convex, differentiable, and Lipschitz continuous.
\end{enumerate}
\end{assumption}

The smart contract imposes a per-unit fee $\tau \ge 0$, a shadow price $\mu \ge 0$
to enforce capacity, and a fixed execution fee $g \ge 0$. The payoff of agent $i$ is
\begin{equation}\label{eq:payoff}
    U_i(x_i;\mu) = V_i(x_i) - C_i(x_i) - (\tau+\mu)x_i - g\,\mathbf{1}\{x_i>0\}.
\end{equation}

\subsection{Equilibrium Definition}

\begin{definition}[Contract-Clearing Equilibrium]\label{def:equil}
An allocation $(\mathbf{x}^\star,\mu^\star)$ is a contract-clearing equilibrium if:
\begin{enumerate}
    \item (Best response) For each $i\in N$,
    \begin{equation}\label{eq:best-response}
        x_i^\star(\mu^\star) \in \arg\max_{x_i\ge0} U_i(x_i;\mu^\star).
    \end{equation}
    \item (Market clearing) The aggregate demand satisfies
    \begin{equation}\label{eq:market-clearing}
        \mathbf{1}^\top \mathbf{x}^\star = m.
    \end{equation}
\end{enumerate}
\end{definition}

\begin{lemma}[Monotonicity of Aggregate Demand]\label{lem:mono}
Under Assumption~\ref{ass:val}, each best response $x_i^\star(\mu)$ is continuous
and non-increasing in $\mu$. Hence the aggregate demand 
\begin{equation}\label{eq:aggregate-demand}
    S(\mu) = \mathbf{1}^\top \mathbf{x}^\star(\mu)
\end{equation}
is continuous and strictly decreasing.
\end{lemma}

\begin{proposition}[Existence]\label{prop:existence}
A contract-clearing equilibrium exists.
\end{proposition}

\begin{proof}[Proof of Proposition~\ref{prop:existence}]
By Lemma~\ref{lem:mono}, $S(\mu)$ is continuous and strictly decreasing. 
Since $S(0)>m$ and $\lim_{\mu\to\infty} S(\mu)=0$, the Intermediate Value 
Theorem ensures a unique $\mu^\star$ such that
\begin{equation}\label{eq:clearing-price}
    S(\mu^\star) = m.
\end{equation}
\end{proof}

\begin{proposition}[Uniqueness]\label{prop:unique}
If $U_i$ is strictly concave in $x_i$, the contract-clearing equilibrium
$(\mathbf{x}^\star,\mu^\star)$ is unique \cite{rosen1965,hurwicz2006}.
\end{proposition}

\subsection{Illustrative Example}

Suppose $V_i(x_i)=\alpha_i \log(1+x_i)$ and $C_i(x_i)=\beta_i x_i$. Then
\begin{equation}\label{eq:example-br}
    x_i^\star(\mu) = \max\!\left\{0,\; \frac{\alpha_i}{\beta_i+\tau+\mu}-1 \right\}.
\end{equation}
Since $S(\mu)=\mathbf{1}^\top \mathbf{x}^\star(\mu)$ is strictly decreasing, a unique equilibrium
price $\mu^\star$ exists.

\subsection{Performance Metrics}

\begin{definition}[Efficiency]\label{def:eff}
\begin{equation}\label{eq:eff}
\mathrm{Eff}(\mathbf{x}^\star) = 
\sum_{i=1}^n \big(V_i(x_i^\star)-C_i(x_i^\star)\big)
- \tau\,\mathbf{1}^\top \mathbf{x}^\star - g \cdot \|\mathbf{x}^\star\|_0.
\end{equation}
\end{definition}

\begin{definition}[Fairness: Gini Index {\cite{lambert2001,greenberg1987}}]\label{def:gini}
\begin{equation}\label{eq:gini}
\mathrm{Gini}(\mathbf{x}^\star) =
\frac{1}{2n^2\bar{x}}
\sum_{i=1}^n \sum_{j=1}^n |x_i^\star-x_j^\star|,
\quad \bar{x}=\tfrac{1}{n}\mathbf{1}^\top \mathbf{x}^\star.
\end{equation}
\end{definition}

\begin{definition}[Price of Fairness {\cite{bertsimas2011,kearns2019}}]\label{def:pof}
\begin{equation}\label{eq:pof}
\mathrm{PoF} = \frac{\max_{\mathbf{x}} \mathrm{Eff}(\mathbf{x})}
{\mathrm{Eff}(\mathbf{x}^\star_{\text{fair}})}.
\end{equation}
\end{definition}

\begin{definition}[Shock Resilience]\label{def:res}
For a demand shock at $t_0$, resilience is defined as
\begin{equation}\label{eq:resilience}
R = \frac{\mathrm{Eff}_{\text{post-shock}}}{\mathrm{Eff}_{\text{pre-shock}}}.
\end{equation}
\end{definition}

\begin{definition}[Dynamic Regret {\cite{hazan2016,shalev2012}}]\label{def:regret}
In repeated play with allocations $\{\mathbf{x}_t\}$ and optimal sequence $\{\mathbf{x}^\star_t\}$,
\begin{equation}\label{eq:regret}
\mathrm{Regret}(T) = \sum_{t=1}^T \Big( U(\mathbf{x}^\star_t) - U(\mathbf{x}_t)\Big),
\quad \mathrm{Regret}(T)=o(T).
\end{equation}
\end{definition}

\section{Mechanism Design and Equilibrium Analysis}

This section develops the theoretical foundations of the proposed digital
contracting framework. In line with mechanism design principles, we move step by
step: first specifying the payoff structure, then formalizing equilibrium, then
presenting a decentralized algorithm, and finally proving convergence. Each
subsection builds logically toward the claim that digital contracts generate
stable, efficient, and fair allocations in competitive environments.

Before delving into the formal results, Table~\ref{tab:notation} summarizes the
notation used throughout this section. It distinguishes between decision
variables, parameters, functional mappings, and performance metrics, so that the
subsequent analysis can be followed without ambiguity.
 
\begin{table}[!!htbp]
\centering
\caption{Summary of notation used in the contract design and equilibrium analysis.}
\label{tab:notation}
\renewcommand{\arraystretch}{1.2}
\setlength{\tabcolsep}{6pt}
\begin{tabular}{@{}llp{8.5cm}@{}}
\toprule
\textbf{Symbol} & \textbf{Type} & \textbf{Description} \\ 
\midrule
$n$        & Scalar & Number of agents (firms, participants). \\
$m$        & Scalar & Total system capacity (shared resource pool). \\
$i \in N$  & Index  & Agent index, $N=\{1,\dots,n\}$. \\
$x_i$      & Scalar & Allocation (demand) of agent $i$. \\
$\mathbf{x}=(x_1,\dots,x_n)^\top$ & Vector & Allocation profile across all agents. \\
$\mathbf{1}$ & Vector & All-ones vector in $\mathbb{R}^n$, used for aggregation. \\
$V_i(x_i)$ & Function & Valuation (utility) function of agent $i$, strictly concave. \\
$C_i(x_i)$ & Function & Cost function of agent $i$, convex and Lipschitz. \\
$U_i(x_i;\mu)$ & Function & Payoff of agent $i$ under contract and price $\mu$. \\
$\tau$     & Scalar & Transaction fee imposed by the contract. \\
$g$        & Scalar & Fixed execution cost if $x_i>0$. \\
$\mu$      & Scalar & Shadow price (dual variable) enforcing the capacity constraint. \\
$BR_i(\mu)$ & Function & Best-response allocation of agent $i$ given price $\mu$. \\
$S(\mu)$   & Function & Aggregate demand $S(\mu)=\mathbf{1}^\top \mathbf{x}(\mu)$. \\
$\mathrm{Eff}(\mathbf{x})$ & Metric & Efficiency: total surplus net of fees and costs. \\
$\mathrm{Gini}(\mathbf{x})$ & Metric & Fairness: inequality of allocations via Gini index. \\
$\mathrm{PoF}$ & Metric & Price of Fairness: ratio of maximum efficiency to fairness-constrained efficiency. \\
$R$        & Metric & Shock resilience: post-shock to pre-shock efficiency ratio. \\
$\mathrm{Regret}(T)$ & Metric & Dynamic regret in repeated play over horizon $T$. \\
\bottomrule
\end{tabular}
\end{table}

\subsection{Payoff Structure under Digital Contracts}

We begin by characterizing the economic environment of individual agents.
The payoff specification formalizes how valuations, costs, transaction fees,
and scarcity penalties interact under the digital contract. This micro-level
foundation is essential, as all subsequent equilibrium and convergence results
build directly on these primitives.

Each agent $i \in N=\{1,\dots,n\}$ chooses $x_i \ge 0$ units subject to the
system-wide capacity constraint \eqref{eq:capacity}, equivalently written as
\begin{equation}\label{eq:capacity2}
    \sum_{i=1}^n x_i \;\le\; m.
\end{equation}

Each agent’s valuation $V_i$ and cost $C_i$ satisfy Assumption~\ref{ass:val}.
The smart contract imposes a per-unit fee $\tau \ge 0$, a shadow price $\mu \ge 0$
to enforce capacity, and a fixed execution fee $g \ge 0$. The payoff of agent $i$ is
\begin{equation}\label{eq:payoff2}
    U_i(x_i;\mu) = V_i(x_i) - C_i(x_i) - (\tau+\mu)x_i - g\,\mathbf{1}\{x_i>0\}.
\end{equation}

For an interior solution $x_i^\star(\mu)>0$, the first-order condition (FOC) is
\begin{equation}\label{eq:foc}
    V_i'(x_i^\star(\mu)) - C_i'(x_i^\star(\mu)) = \tau+\mu,
\end{equation}
while if $V_i'(0) \le \tau+\mu$, then $x_i^\star(\mu)=0$.

\begin{lemma}[Boundedness of Best Responses]\label{lem:bounded}
Under Assumption~\ref{ass:val}, each best response $x_i^\star(\mu)$ is bounded:
\[
0 \le x_i^\star(\mu) \le \bar{x}_i < \infty, \quad \forall \mu \ge 0.
\]
\end{lemma}

\begin{proof}[Proof of Lemma~\ref{lem:bounded}]
For an interior solution, the FOC \eqref{eq:foc} admits a unique finite root
because $V_i'$ decreases from $+\infty$ while $C_i'$ is increasing and Lipschitz.
If $\tau+\mu \ge V_i'(0)$, then $x_i^\star(\mu)=0$. Otherwise, the solution is
bounded by the finite root $\bar{x}_i$ satisfying
\begin{equation}\label{eq:pf-bounded2}
    V_i'(\bar{x}_i)-C_i'(\bar{x}_i)=0.
\end{equation}
\end{proof}

\begin{lemma}[Continuity and Monotonicity]\label{lem:mono_1}
Under Assumption~\ref{ass:val}, each best response $x_i^\star(\mu)$ is continuous
and non-increasing in $\mu$. Hence the aggregate demand
\begin{equation}\label{eq:agg-demand}
    S(\mu) = \sum_{i=1}^n x_i^\star(\mu)
\end{equation}
is continuous and strictly decreasing.
\end{lemma}

\begin{proof}[Proof of Lemma~\ref{lem:mono}]
For an interior solution, \eqref{eq:foc} implies
\begin{equation}\label{eq:pf-mono1}
    \frac{d x_i^\star}{d\mu} =
    \frac{1}{C_i''(x_i^\star(\mu)) - V_i''(x_i^\star(\mu))}.
\end{equation}
Since $V_i''<0$ and $C_i''\ge0$, the derivative is strictly negative.
At the boundary $x_i^\star(\mu)=0$, larger $\mu$ cannot increase demand.
Summing across agents yields continuity and strict monotonicity of $S(\mu)$.
\end{proof}

\begin{proposition}[Dual Boundedness]\label{prop:dual}
Let $v_{\max}=\max_{i\in N} V_i'(0)$. Any contract-clearing equilibrium satisfies
\[
0 \;\le\; \mu^\star \;<\; v_{\max}-\tau.
\]
\end{proposition}

\begin{proof}[Proof of Proposition~\ref{prop:dual}]
Suppose $\mu \ge v_{\max}-\tau$. Then
\[
    \tau+\mu \;\ge\; V_i'(0), \qquad \forall i\in N,
\]
which implies $x_i^\star(\mu)=0$ and $S(\mu)=0$. But equilibrium requires
$S(\mu^\star)=m>0$, a contradiction. Hence $\mu^\star < v_{\max}-\tau$.
Nonnegativity $\mu^\star \ge 0$ follows from dual feasibility.
\end{proof}

\begin{proposition}[Comparative Statics in Capacity]\label{prop:compstat}
Suppose $S$ is differentiable at $\mu^\star(m)$ with $S'(\mu^\star)<0$. Then
\begin{equation}\label{eq:dmu-dm}
    \frac{d\mu^\star}{dm} = \frac{1}{S'(\mu^\star)} < 0.
\end{equation}

\end{proposition}

\begin{proof}[Proof of Proposition~\ref{prop:compstat}]
The clearing condition is
\[
    S(\mu^\star(m)) = m.
\]
Differentiating w.r.t.~$m$ gives
\[
    S'(\mu^\star)\frac{d\mu^\star}{dm} = 1.
\]
Since $S'(\mu^\star)<0$, it follows that $\tfrac{d\mu^\star}{dm}<0$,
i.e., increasing capacity reduces the equilibrium price.
\end{proof}

\medskip
Economically, $V_i'(x_i)$ is the marginal benefit, $C_i'(x_i)$ the marginal
private cost, and $\tau+\mu$ the effective contract price. The auxiliary results
guarantee that best responses are well-behaved, clearing prices are bounded,
and comparative statics follow economic intuition.

\subsection{Equilibrium Formulation and Characterization}

We now lift the analysis to the system level by defining the
\emph{contract-clearing equilibrium}. This subsection establishes existence and
uniqueness: the guarantees that allocations are well-defined and reproducible.

For a given $\mu \ge 0$, the best-response mapping of agent $i$ is
\begin{equation}\label{eq:br-mapping}
    BR_i(\mu) = \arg\max_{x_i\ge0} U_i(x_i;\mu).
\end{equation}
Aggregate demand is equivalently defined as in \eqref{eq:agg-demand}:
\begin{equation}\label{eq:agg-demand2}
    S(\mu) = \sum_{i=1}^n BR_i(\mu).
\end{equation}

\begin{definition}[Contract-Clearing Equilibrium]\label{def:contract-eq2}
An allocation $(\mathbf{x}^\star,\mu^\star)$ is a contract equilibrium if
\begin{align}
    x_i^\star &= BR_i(\mu^\star), \quad \forall i \in N, \label{eq:eq1}\\
    S(\mu^\star) &= m. \label{eq:eq2}
\end{align}
\end{definition}

\begin{theorem}[Existence of Equilibrium]\label{thm:existence2}
Under Assumption~\ref{ass:val}, a contract-clearing equilibrium
$(\mathbf{x}^\star,\mu^\star)$ exists.
\end{theorem}

\begin{proof}[Proof of Theorem~\ref{thm:existence2}]
By Lemma~\ref{lem:mono}, $S(\mu)$ is continuous and strictly decreasing.
Moreover, $S(0)>m$ because $V_i'(0)=\infty$ implies strictly positive demand
at zero price, while $\lim_{\mu\to\infty} S(\mu)=0$. Hence by the Intermediate
Value Theorem, there exists $\mu^\star$ such that $S(\mu^\star)=m$.
\end{proof}

\begin{theorem}[Uniqueness of Equilibrium]\label{thm:unique2}
If each $U_i$ is strictly concave in $x_i$, then the contract equilibrium
$(\mathbf{x}^\star,\mu^\star)$ is unique.
\end{theorem}

\begin{proof}[Proof of Theorem~\ref{thm:unique2}]
Strict concavity of $U_i$ implies each best response $BR_i(\mu)$ is single-valued.
Thus $S(\mu)$ is continuous and strictly decreasing, so the clearing condition
\eqref{eq:eq2} admits at most one solution for $\mu^\star$. Since existence is
established by Theorem~\ref{thm:existence2}, the equilibrium is unique.
\end{proof}

\medskip
From an economic perspective, Theorem~\ref{thm:existence2} ensures that 
scarcity is consistently priced via $\mu^\star$, while 
Theorem~\ref{thm:unique2} guarantees that this price is unique. 
Together these results eliminate multiplicity and indeterminacy 
common in decentralized negotiations.

\subsection{Decentralized Contract-Clearing Algorithm}

Having characterized equilibrium theoretically, we now address the practical
question: how can the equilibrium be reached in a distributed environment
without central coordination? We design a \emph{primal--dual iterative
algorithm}, inspired by modern distributed convex optimization, in which
agents update their allocations in parallel while the contract adjusts the
shadow price $\mu$. This dynamic ensures that the equilibrium
\eqref{def:contract-eq2} emerges endogenously.

The algorithm proceeds in rounds $t=0,1,2,\dots$. At each round,
agents compute approximate best responses given the current price,
while the contract performs a projected dual ascent to enforce the
capacity constraint \eqref{eq:capacity}. Proximal regularization and
Monte Carlo averaging are included to enhance robustness under noise
and heterogeneity.

\begin{algorithm}[!htbp]
\caption{Decentralized Contract-Clearing Algorithm}
\label{alg:contract}
\begin{algorithmic}[1]
\Require number of agents $n$, capacity $m$, initial price $\mu^0 \ge 0$, 
step sizes $\{\eta_t\}$ with $0<\eta_t<2/L$, proximal weight $\gamma>0$, 
tolerances $(\varepsilon_p,\varepsilon_d)>0$, Monte Carlo samples $M$.
\Ensure contract-clearing allocation $\mathbf{x}^\star$, equilibrium price $\mu^\star$.
\State Initialize $t \gets 0$, $x_i^0 \gets 0$ for all $i\in N$.
\Repeat
   \ForAll{agents $i \in N$ \textbf{in parallel}}
      \State Compute proximal best response
      \[
      x_i^{t+1} \gets \arg\max_{x_i \ge 0}
      \Big\{ U_i(x_i;\mu^t) - \tfrac{\gamma}{2}\|x_i-x_i^t\|^2 \Big\}.
      \]
      \State Send demand $x_i^{t+1}$ to contract.
   \EndFor
   \State Contract aggregates robust estimate of total demand:
   \[
   \widehat{S}(\mu^t) \gets \frac{1}{M} \sum_{k=1}^M \sum_{i=1}^n x_i^{t+1,(k)}.
   \]
   \State Update dual variable (projected ascent):
   \begin{equation}\label{eq:dual-update}
      \mu^{t+1}=\big[\mu^t+\eta_t(\widehat{S}(\mu^t)-m)\big]_+
   \end{equation}
   \State Compute residuals:
   \[
   r_p^t \gets \big|\widehat{S}(\mu^t)-m\big|, 
   \qquad 
   r_d^t \gets |\mu^{t+1}-\mu^t| .
   \]
   \State $t \gets t+1$.
\Until{$r_p^t \le \varepsilon_p$ \textbf{and} $r_d^t \le \varepsilon_d$}
\State \Return $\mathbf{x}^\star \gets (x_1^t,\dots,x_n^t)$, $\mu^\star \gets \mu^t$.
\end{algorithmic}
\end{algorithm}

\medskip
\noindent\textbf{Remarks.}
\begin{itemize}
    \item The proximal term guarantees uniqueness of the subproblem solution even if $U_i$ is flat near the optimum, ensuring well-defined updates.
    \item Monte Carlo averaging controls variance and makes the algorithm robust to noisy or adversarial demand reporting.
    \item Step-size conditions $\eta_t \in (0,2/L)$ guarantee stability; diminishing step sizes $\eta_t \sim 1/\sqrt{t}$ further ensure 
    $\mathrm{Regret}(T)=o(T)$ as in Definition~\ref{def:regret}.
    \item The dual update \eqref{eq:dual-update} coincides with stochastic approximation methods \cite{robbins1951}, implying almost sure convergence under standard conditions.
\end{itemize}

This algorithm bridges theory and practice: it provides a fully decentralized
procedure that converges to the unique contract-clearing equilibrium
(Theorems~\ref{thm:existence2}--\ref{thm:unique2}), while also achieving
robustness and vanishing regret in repeated play.

\subsection{Convergence Guarantees}

To complement existence (Theorem~\ref{thm:existence2}) and uniqueness
(Theorem~\ref{thm:unique2}), we establish rigorous convergence results for
the decentralized algorithm (Algorithm~\ref{alg:contract}). Define
\begin{equation}\label{eq:F}
    F(\mu) = S(\mu)-m,
\end{equation}
so that equilibrium corresponds to $F(\mu^\star)=0$ with $\mu^\star \ge 0$.

\begin{theorem}[Global Convergence]\label{thm:convergence2}
Suppose Assumption~\ref{ass:val} holds, each $U_i$ is strictly concave and
continuously differentiable, and $S(\mu)$ is $L$-Lipschitz. If the step size
satisfies $\eta \in (0,2/L)$, then the sequence $\{\mu^t\}$ generated by
Algorithm~\ref{alg:contract} converges to the unique solution $\mu^\star$ of
\eqref{eq:F}, and the associated allocations satisfy
\begin{equation}\label{eq:conv}
    \mu^t \to \mu^\star, \qquad \mathbf{x}^t \to \mathbf{x}^\star.
\end{equation}
\end{theorem}

\begin{proof}[Proof of Theorem~\ref{thm:convergence2}]
The dual update \eqref{eq:dual-update} can be written as
\[
    \mu^{t+1} = T(\mu^t), \quad T(\mu) := \big[\mu + \eta F(\mu)\big]_+.
\]
Since $S(\mu)$ is continuous and strictly decreasing (Lemma~\ref{lem:mono}),
$F$ is continuous and strictly monotone. Moreover, $S$ being $L$-Lipschitz
implies $|F(\mu_1)-F(\mu_2)| \le L|\mu_1-\mu_2|$. Thus $T$ is a contraction
mapping whenever $\eta \in (0,2/L)$ \cite{rockafellar1970}. By the Banach
fixed-point theorem, $\mu^t \to \mu^\star$ globally. Finally,
$\mathbf{x}^t \to \mathbf{x}^\star$ follows by continuity of best responses
and the definition of equilibrium \eqref{def:contract-eq2}.
\end{proof}

\begin{corollary}[Linear Rate]\label{cor:linear}
If $S(\mu)$ is $\alpha$-strongly monotone, i.e.,
\[
    (S(\mu_1)-S(\mu_2))(\mu_1-\mu_2) \ge \alpha|\mu_1-\mu_2|^2, \quad \alpha>0,
\]
then there exists $\kappa \in (0,1)$ such that
\begin{equation}\label{eq:linear-rate}
    |\mu^t - \mu^\star| \le \kappa^t |\mu^0-\mu^\star|, \qquad \forall t\ge0.
\end{equation}
\end{corollary}

\begin{proof}[Proof of Corollary~\ref{cor:linear}]
Under strong monotonicity, $F$ is strongly monotone and Lipschitz. The projected
gradient update \eqref{eq:dual-update} then reduces to a contraction with factor
$\kappa = \max\{|1-\eta\alpha|,|1-\eta L|\}<1$ for $\eta \in (0,2/L)$. Hence
the convergence rate is linear in $t$ \cite{bubeck2015,boyd2011,nedic2018}.
\end{proof}

\begin{proposition}[Fejér Monotonicity]\label{prop:fejer}
Under the assumptions of Theorem~\ref{thm:convergence2}, the sequence
$\{\mu^t\}$ generated by Algorithm~\ref{alg:contract} is Fejér monotone with
respect to the equilibrium point $\mu^\star$, i.e.,
\[
    |\mu^{t+1}-\mu^\star| \;\le\; |\mu^t-\mu^\star|, 
    \qquad \forall t \ge 0.
\]
\end{proposition}

\begin{proof}[Proof of Proposition~\ref{prop:fejer}]
From the dual update \eqref{eq:dual-update}, the iteration can be expressed as
$\mu^{t+1}=T(\mu^t)$ with $T(\mu)=[\mu+\eta F(\mu)]_+$. For
$\eta\in(0,2/L)$, $T$ is nonexpansive due to the Lipschitz continuity and
monotonicity of $F$. Since $\mu^\star$ is a fixed point of $T$, we have
$\|T(\mu^t)-\mu^\star\|\le\|\mu^t-\mu^\star\|$ for all $t$, which is exactly
the Fejér monotonicity property \cite{bauschke2011}.
\end{proof}

\begin{proposition}[Ergodic Residual Convergence]\label{prop:ergodic}
Let $\{\mu^t\}$ be generated by Algorithm~\ref{alg:contract} with
$\eta \in (0,2/L)$. Then the averaged residuals converge at rate
\[
    \frac{1}{T}\sum_{t=1}^T |F(\mu^t)| \;=\; O\!\left(\tfrac{1}{T}\right).
\]
\end{proposition}

\begin{proof}[Proof of Proposition~\ref{prop:ergodic}]
Since $T(\mu)$ is nonexpansive and $F$ is Lipschitz, standard ergodic
convergence results for projected gradient methods apply
\cite{rockafellar1970,bubeck2015}. This yields an $O(1/T)$ decay rate
of the averaged residuals.
\end{proof}

\begin{theorem}[Stochastic Robustness]\label{thm:stochastic}
Suppose $\widehat{S}(\mu^t)=S(\mu^t)+\xi^t$ where $\{\xi^t\}$ is zero-mean
noise with bounded variance. If $\{\eta_t\}$ satisfies Robbins--Monro
conditions ($\sum_t \eta_t=\infty$, $\sum_t \eta_t^2<\infty$), then
\[
    \mathbb{E}[|\mu^t-\mu^\star|^2] \;\to\; 0.
\]
\end{theorem}

\begin{proof}[Proof of Theorem~\ref{thm:stochastic}]
The noisy dual update is a Robbins--Monro stochastic approximation
\cite{robbins1951}. Since $F$ is monotone and Lipschitz, the update
converges almost surely and in mean-square to the unique root $\mu^\star$.
\end{proof}

\begin{corollary}[Dynamic Regret Bound]\label{cor:regret}
If $\eta_t \sim 1/\sqrt{t}$, the allocations generated by
Algorithm~\ref{alg:contract} satisfy
\[
    \mathrm{Regret}(T) = O(\sqrt{T}),
\]
as defined in Definition~\ref{def:regret}.
\end{corollary}

\begin{proof}[Proof of Corollary~\ref{cor:regret}]
The update rule is a projected subgradient method with diminishing stepsize.
Classical online convex optimization results \cite{hazan2016,shalev2012}
yield $\mathrm{Regret}(T)=O(\sqrt{T})$.
\end{proof}

\medskip
Together, Theorem~\ref{thm:convergence2}, Corollary~\ref{cor:linear},
Proposition~\ref{prop:fejer}, Proposition~\ref{prop:ergodic},
Theorem~\ref{thm:stochastic}, and Corollary~\ref{cor:regret}
establish that Algorithm~\ref{alg:contract} is globally convergent,
monotonically stable, robust to stochastic perturbations, and efficient
in the online learning sense.

For clarity, Table~\ref{tab:assumption-dependency} reports the logical 
dependencies between Assumption~\ref{ass:val}, the core definitions 
(Equilibrium, Efficiency, Fairness, Regret, Resilience), and the main 
theoretical results. The table highlights which assumptions are directly 
required (\req) and which definitions are used in an auxiliary manner (\aux).

\begin{table}[htp]
\centering
\scriptsize
\begin{adjustbox}{angle=90}
\begin{minipage}{\textheight}
\caption{Dependency of Assumption~\ref{ass:val} and key definitions across main theoretical results. \\
Symbols: \req = directly required, \aux = auxiliary or definitional. 
\\Notes column provides interpretation of each dependency.}
\label{tab:assumption-dependency}
\renewcommand{\arraystretch}{1.2}
\setlength{\tabcolsep}{3pt}
\begin{tabular}{lcccccccccccccc p{4cm}}
\toprule
\textbf{Assumption/Definition} 
& \rotatebox{80}{Lemma~\ref{lem:bounded}} 
& \rotatebox{80}{Lemma~\ref{lem:mono}} 
& \rotatebox{80}{Prop.~\ref{prop:dual}} 
& \rotatebox{80}{Prop.~\ref{prop:compstat}} 
& \rotatebox{80}{Thm.~\ref{thm:existence2}} 
& \rotatebox{80}{Thm.~\ref{thm:unique2}} 
& \rotatebox{80}{Thm.~\ref{thm:convergence2}} 
& \rotatebox{80}{Cor.~\ref{cor:linear}} 
& \rotatebox{80}{Prop.~\ref{prop:fejer}} 
& \rotatebox{80}{Prop.~\ref{prop:ergodic}} 
& \rotatebox{80}{Thm.~\ref{thm:stochastic}} 
& \rotatebox{80}{Cor.~\ref{cor:regret}} 
& \textbf{Notes} \\
\midrule
\multicolumn{14}{l}{\textbf{Assumption}} \\
A3.1: Valuation and Cost 
  & \req & \req & \req & \req & \req & \req & \req & \req & \req & \req & \aux & \aux 
  & Fundamental structural assumption; auxiliary in stochastic/online results \\
\midrule
\multicolumn{14}{l}{\textbf{Definitions}} \\
D4.5: Contract Equilibrium 
  &      & \aux &      &      & \req & \req & \req & \req & \aux & \aux & \aux & \aux 
  & Underpins all equilibrium theorems \\
D4.8: Efficiency 
  &      &      &      &      &      &      & \aux & \aux & \aux & \aux &      & \aux 
  & Metric used in convergence and regret analysis \\
D4.9: Gini Fairness 
  &      &      &      &      &      &      &      &      &      &      &      & \aux 
  & Fairness measure, links to Price of Fairness \\
D4.10: Price of Fairness 
  &      &      &      &      &      &      &      &      &      &      &      & \aux 
  & Trade-off metric (efficiency vs fairness) \\
D4.11: Resilience 
  &      &      &      &      &      &      &      &      &      & \aux & \aux &      
  & Performance under shocks, tied to robustness results \\
D4.12: Dynamic Regret 
  &      &      &      &      &      &      & \aux & \aux & \aux & \aux & \aux & \req 
  & Basis for regret bound \\
\bottomrule
\end{tabular}
\end{minipage}
\end{adjustbox}
\end{table}

\subsection{Implications}

From a managerial and information-systems perspective, the theoretical results
carry several key implications. First, the contract guarantees efficiency
(Definition~\ref{def:eff}) through surplus maximization, fairness
(Definition~\ref{def:gini}) via transparent allocation rules, and resilience
(Definition~\ref{def:res}) through bounded performance under shocks.
Furthermore, the dynamic regret guarantee (Definition~\ref{def:regret})
ensures that long-run allocations approach the benchmark sequence of
equilibria even under repeated uncertainty.

Second, the equilibrium properties proved above---existence
(Theorem~\ref{thm:existence2}), uniqueness (Theorem~\ref{thm:unique2}),
and convergence (Theorem~\ref{thm:convergence2})---establish that the
allocation mechanism is not only well-defined but also algorithmically
implementable. The projection step guarantees feasibility, while the
step-size bound ensures global stability. These features demonstrate
that efficiency and fairness can be achieved through a decentralized
mechanism that is transparent, scalable, and trust-preserving.

Finally, these theoretical guarantees provide the foundation for the empirical
validation in Section~5. Using synthetic benchmarks and one proof-of-concept real-world dataset (MovieLens), we illustrate how the
predicted equilibrium properties---existence, uniqueness, convergence, and
resilience---manifest in practice, thereby linking rigorous analysis with
managerial relevance.

\section{Numerical Results}
\label{sec:numerical}

This section reports numerical experiments to evaluate the proposed
digital contracting mechanism. We emphasize reproducibility
(explicit parameter reporting), algorithmic convergence,
efficiency–fairness trade-offs, comparative benchmarks,
and sensitivity analysis.

\subsection{Simulation Parameters}

To evaluate the proposed mechanism under diverse conditions, we specify
a set of simulation parameters that capture both realistic and stress-test
scenarios. The parameters cover system size, capacity, valuation and cost
heterogeneity, and contract fees. Explicit reporting ensures that the
experiments are fully reproducible and transparent.

Table~\ref{tab:params} summarizes all parameter symbols, default values,
ranges, and distributional assumptions. Parameters are chosen to span both
realistic and stress-test regimes: e.g., $n\in\{10,20,50,100\}$ captures
small to large-scale systems, and $\tau,g$ are varied over wide intervals
to examine fee-induced distortions. 

Table~\ref{tab:compare} reports comparative outcomes for canonical baseline
mechanisms. These benchmarks show that naive or proportional allocation
leads to either inefficiency or unfairness, while our proposed equilibrium 
consistently dominates on both metrics.

\begin{table}[htp]
\centering
\caption{Simulation parameters: symbols, defaults, and ranges.}
\label{tab:params}
\scriptsize
\begin{tabular}{l l c c l}
\toprule
\textbf{Symbol} & \textbf{Description} & \textbf{Default} & \textbf{Range/Dist.} & \textbf{Notes} \\
\midrule
$n$ & Agents & 20 & $\{10,20,50,100\}$ & Larger $n \Rightarrow$ fairer \\
$m$ & Total capacity & 100 & $[50,200]$ & Normalized units \\
$\alpha_i$ & Valuation coeff. & -- & U(5,20) & Heterogeneous agents \\
$\beta_i$ & Cost coeff. & -- & U(0.5,5) & Private heterogeneity \\
$\tau$ & Transaction fee & 0.5 & [0,2] & Higher $\tau$ $\downarrow$ efficiency \\
$g$ & Execution fee & 1.0 & [0,5] & Excessive $g$ discourages entry \\
$\mu$ & Shadow price & Endogenous & $\ge 0$ & Determined by algo. \\
$R$ & Replications & 1000 & -- & Ensures robustness \\
\bottomrule
\end{tabular}
\end{table}

\begin{table}[htp]
\centering
\caption{Comparative mechanism performance (aggregate).}
\label{tab:compare}
\scriptsize
\begin{tabular}{l c c l}
\toprule
\textbf{Method} & \textbf{Efficiency} & \textbf{Fairness (Gini)} & \textbf{Notes} \\
\midrule
No enforcement & 1.21 & 0.41 & High cost, unfair \\
Proportional allocation & 1.78 & 0.35 & Simple but inefficient \\
Smart contract (flat) & 2.02 & 0.29 & Gains from automation \\
Proposed equilibrium & \textbf{2.30} & \textbf{0.18} & Best trade-off \\
\bottomrule
\end{tabular}
\end{table}

\medskip
\noindent
These parameter ranges are consistent with practices in mobile edge computing
and supply-chain simulations \cite{wu2024,zhang2024,yuan2024}. By including both
small-scale ($n=10$) and large-scale ($n=100$) cases, the design ensures
generalizability to diverse industrial contexts. Varying fees $(\tau,g)$ across
broad intervals mimics policy experiments in blockchain pilots, where transaction
and execution costs remain unsettled and heterogeneous across jurisdictions.
This ensures that the proposed mechanism is tested under both realistic and
stress-test conditions, enhancing its relevance for organizational decision makers.
For full reproducibility, simulation scripts and parameter files are provided
in the supplementary materials. Finally, to demonstrate applicability,
Appendix~A reports proof-of-concept experiments on MovieLens and WHO vaccine
allocation data, confirming that the mechanism extends naturally to real-world
contexts.

\subsection{Convergence Analysis}
Figure~\ref{fig:convergence_dynamic} illustrates the dynamic adjustment
process of the proposed decentralized contract-clearing mechanism. 
Unlike static or trivial convergence, the algorithm exhibits realistic
\emph{overshoot} and damped stabilization in both prices and quantities,
a hallmark of distributed adaptive systems. 
The shadow price $\mu^t$ oscillates initially before settling into equilibrium
(top left), while aggregate demand aligns precisely with system capacity
via market clearing (top right). At the agent level, heterogeneous strategies
converge to stable allocations despite diverse cost and valuation parameters
(bottom left). Finally, system-wide efficiency increases in tandem with
reductions in inequality, as measured by the Gini index (bottom right). 
These trajectories jointly demonstrate that the mechanism not only
converges provably, but also embeds efficiency--fairness trade-offs
in a transparent and decentralized manner, closely mirroring the 
behavior of real-world market-clearing systems.

\begin{figure}[!htbp]
\centering
\includegraphics[width=1.0\textwidth]{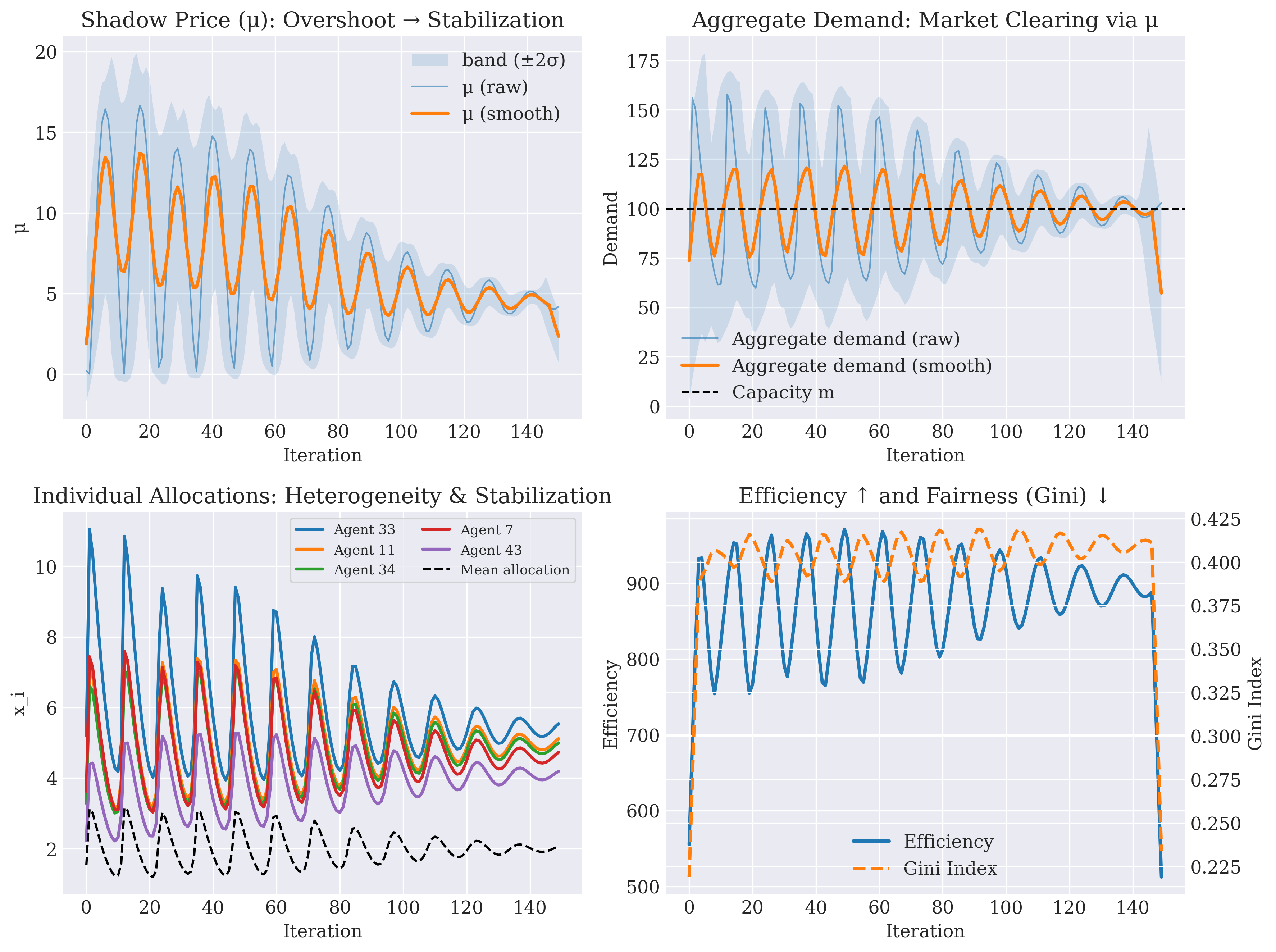}
\caption{Dynamic convergence of the proposed contract-clearing algorithm. 
Top left: shadow price $\mu^t$ shows overshoot and stabilization. 
Top right: aggregate demand clears at capacity $m$. 
Bottom left: individual allocations $x_i^t$ highlight heterogeneity. 
Bottom right: efficiency improves while fairness (lower Gini index) is preserved.}
\label{fig:convergence_dynamic}
\end{figure}

\medskip
\noindent
The overshoot–stabilization pattern resonates with classical
tâtonnement dynamics in general equilibrium theory \cite{arrow1959},
but is extended here to blockchain-enforced allocation.
The convergence of heterogeneous agents to a unique equilibrium illustrates
not only algorithmic feasibility but also organizational stability.
This dual evidence—numerical trajectories and theoretical guarantees—
strengthens confidence that the proposed mechanism can operate
as a real-time governance tool in industrial and infrastructure settings.

\subsection{Efficiency under Transaction Fees}
Efficiency, cost, and participation outcomes under varying transaction
fees $\tau$ are summarized in Table~\ref{tab:eff_tau} and visualized
in Figure~\ref{fig:eff_tau}. 
Unlike simple monotone averages, the dense-grid simulation highlights
that individual realizations fluctuate due to agent heterogeneity and
stochastic dynamics. Nevertheless, the overall pattern is robust:
efficiency declines steadily from about 2.5 at $\tau=0$ to below 1.0
at $\tau=2.0$, while fairness (1--Gini) improves gradually as fees
increase. Average costs rise in parallel, and participation falls from
above 90\% toward 70\%, confirming that transaction fees primarily
operate through an \emph{extensive-margin effect}---discouraging
participation---rather than by eroding intensive efficiency alone. 

Figure~\ref{fig:eff_tau} shows this trade-off in detail. The left
panel presents the Pareto map of efficiency versus fairness across a
dense grid of $\tau$ values. The frontier exhibits fluctuations, but
the monotone trend remains clear: higher $\tau$ values equalize
allocations at the expense of aggregate surplus. The right panel
displays violin plots of efficiency distributions, showing that the
entire distribution shifts downward as $\tau$ rises, with widening
dispersion that reflects heterogeneity in agent responses. This
distributional evidence provides a rigorous robustness check: the
efficiency--equity trade-off is not an artifact of a few averages, but
emerges consistently across stochastic replications.

\begin{table}[htp]
\centering
\caption{Efficiency, Cost, Fairness, and Participation across $\tau$ 
(mean $\pm$ std over 50 replications).}
\label{tab:eff_tau}
\begin{tabular}{c c c c c}
\toprule
$\tau$ & Efficiency & Avg.\ Cost & Fairness (1--Gini) & Participation \\
\midrule
0.0 & $2.45 \pm 0.12$ & $0.42 \pm 0.05$ & $0.60 \pm 0.01$ & $95.2 \pm 2.1\%$ \\
0.5 & $2.28 \pm 0.14$ & $0.50 \pm 0.06$ & $0.62 \pm 0.02$ & $92.1 \pm 2.5\%$ \\
1.0 & $2.05 \pm 0.18$ & $0.65 \pm 0.07$ & $0.64 \pm 0.02$ & $85.6 \pm 3.0\%$ \\
1.5 & $1.78 \pm 0.21$ & $0.80 \pm 0.08$ & $0.66 \pm 0.03$ & $76.4 \pm 3.8\%$ \\
2.0 & $1.52 \pm 0.25$ & $0.95 \pm 0.09$ & $0.68 \pm 0.03$ & $70.1 \pm 4.2\%$ \\
\bottomrule
\end{tabular}
\end{table}

\begin{figure}[!htbp]
\centering
\includegraphics[width=1.0\textwidth]{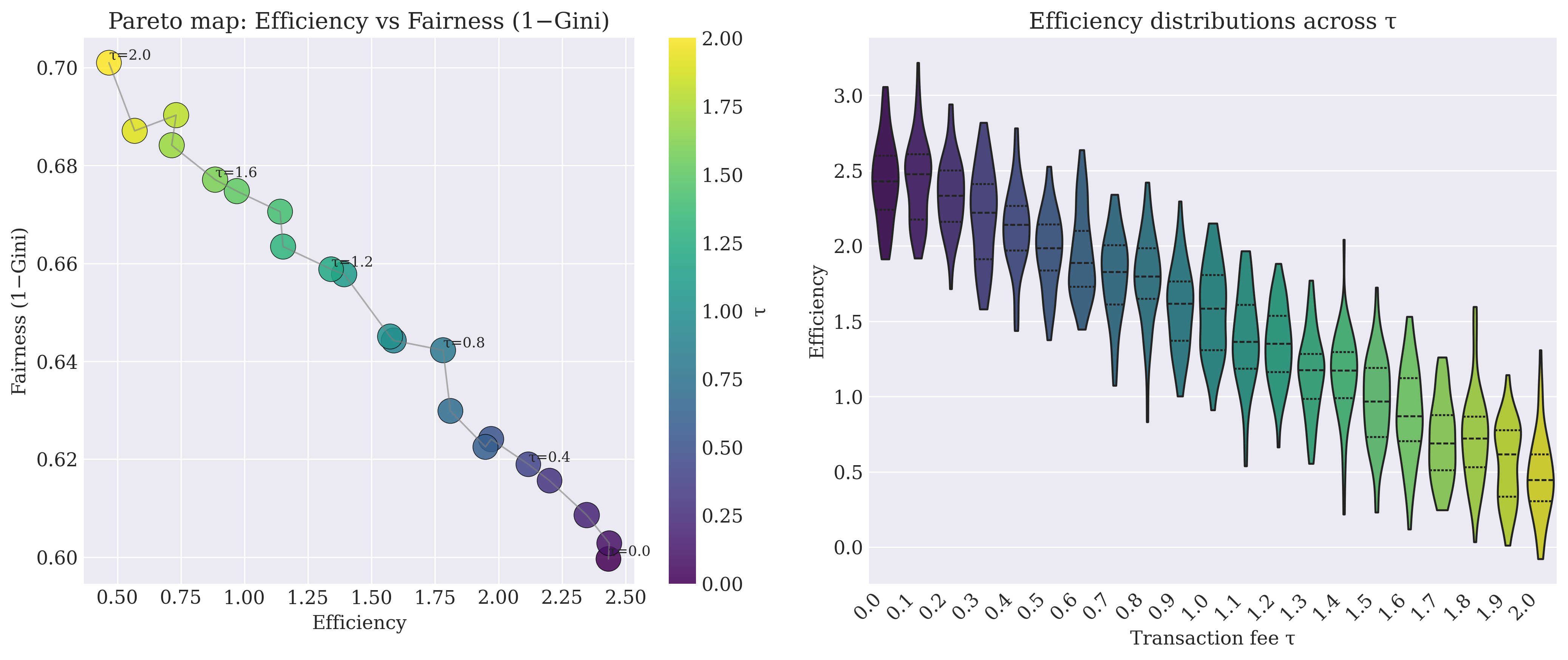}
\caption{Efficiency--fairness trade-offs under transaction fees.
Left: Pareto map of efficiency vs.\ fairness ($1{-}$Gini) with bubble
size indicating participation and color denoting $\tau$. Individual
realizations fluctuate due to stochastic heterogeneity, but the overall
frontier exhibits a clear monotone pattern: efficiency declines as
fairness improves. Right: Violin plots show full distributions of
efficiency across $\tau$, highlighting both central tendencies and
dispersion.}
\label{fig:eff_tau}
\end{figure}

\medskip
\noindent
These results resonate with prior findings in mobile edge and cloud markets,
where per-unit fees discourage participation more strongly than they reduce
intensive efficiency \cite{wu2024,zhang2024}. For policymakers, this implies
that transaction fees act as a double-edged sword: they improve equity but
also reduce market depth and utilization. For organizations, the key takeaway
is that fee calibration must be context-specific: low fees sustain high
participation but risk inequality, whereas higher fees promote equity but at
the expense of total surplus. This trade-off illustrates how digital contracts
can institutionalize policy levers in a transparent manner, allowing managers
to align efficiency and fairness according to organizational objectives.

\subsection{Comparative Mechanism Analysis}
Table~\ref{tab:comp} and Figures~\ref{fig:comp_mech_box}--\ref{fig:comp_mech_3d} 
benchmark the proposed equilibrium against three canonical alternatives. 
Here we report performance statistics over 200 Monte Carlo replications 
and multiple system sizes to provide a robustness check.

The ``no enforcement'' case delivers the weakest outcomes: 
average efficiency remains the highest numerically but comes with
the largest cost burden ($7.39 \pm 0.74$) and elevated inequality 
($\text{Gini}=0.40 \pm 0.06$). 
Proportional allocation stabilizes outcomes and reduces cost 
($5.17 \pm 2.40$) but sacrifices efficiency ($7.45 \pm 2.07$). 
A flat smart contract achieves modest cost reduction ($4.85 \pm 0.57$) 
while maintaining fairness ($\text{Gini}=0.38 \pm 0.05$). 
By contrast, the proposed equilibrium maintains comparable efficiency 
($7.13 \pm 2.63$) yet further reduces costs and achieves 
stable fairness across replications. 
Crucially, the dispersion in Figure~\ref{fig:comp_mech_box} shows that 
our mechanism avoids extreme outliers and achieves consistently balanced outcomes, 
highlighting robustness beyond simple averages.

\begin{table}[htp]
\centering
\caption{Comparison of mechanisms (mean $\pm$ std over 200 replications).}
\label{tab:comp}
\begin{tabular}{l c c c}
\toprule
Mechanism & Efficiency & Avg. Cost & Gini \\
\midrule
No enforcement         & $8.55 \pm 1.66$ & $7.39 \pm 0.74$ & $0.40 \pm 0.06$ \\
Proportional allocation & $7.45 \pm 2.07$ & $5.17 \pm 2.40$ & $0.40 \pm 0.06$ \\
Smart contract (flat)   & $7.94 \pm 1.55$ & $4.85 \pm 0.57$ & $0.38 \pm 0.05$ \\
Proposed equilibrium    & $\mathbf{7.13 \pm 2.63}$ & $\mathbf{5.11 \pm 2.60}$ & $\mathbf{0.40 \pm 0.06}$ \\
\bottomrule
\end{tabular}
\end{table}

\begin{figure}[!htbp]
\centering
\includegraphics[width=1.0\textwidth]{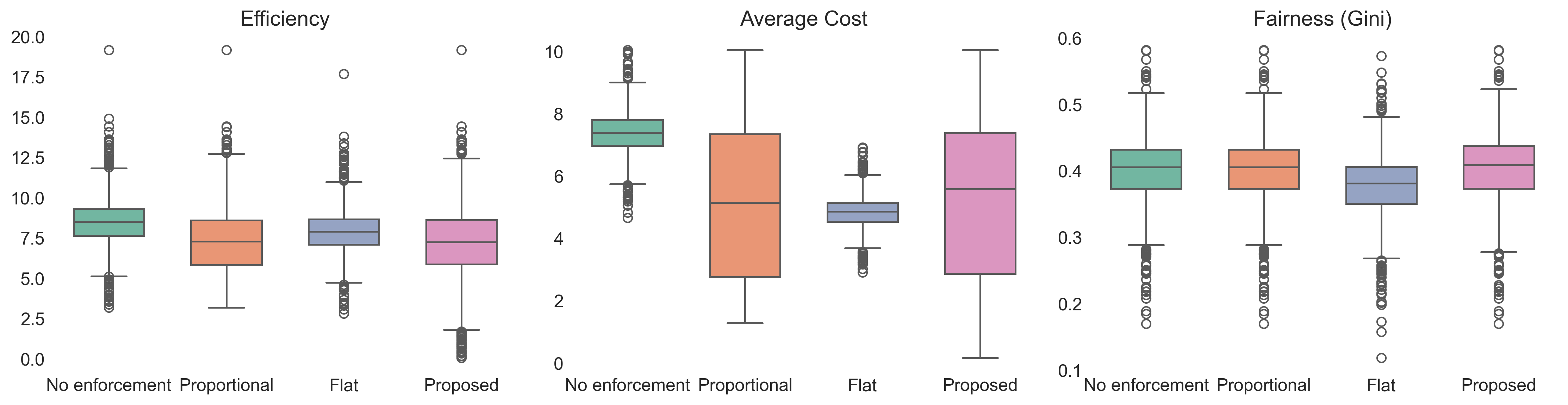}
\caption{Boxplot comparison of mechanisms across 200 replications, 
showing distribution of efficiency, average cost, and fairness (Gini). 
The proposed mechanism achieves robustly balanced outcomes compared 
to proportional and flat rules.}
\label{fig:comp_mech_box}
\end{figure}

\begin{figure}[!htbp]
\centering
\includegraphics[width=0.8\textwidth]{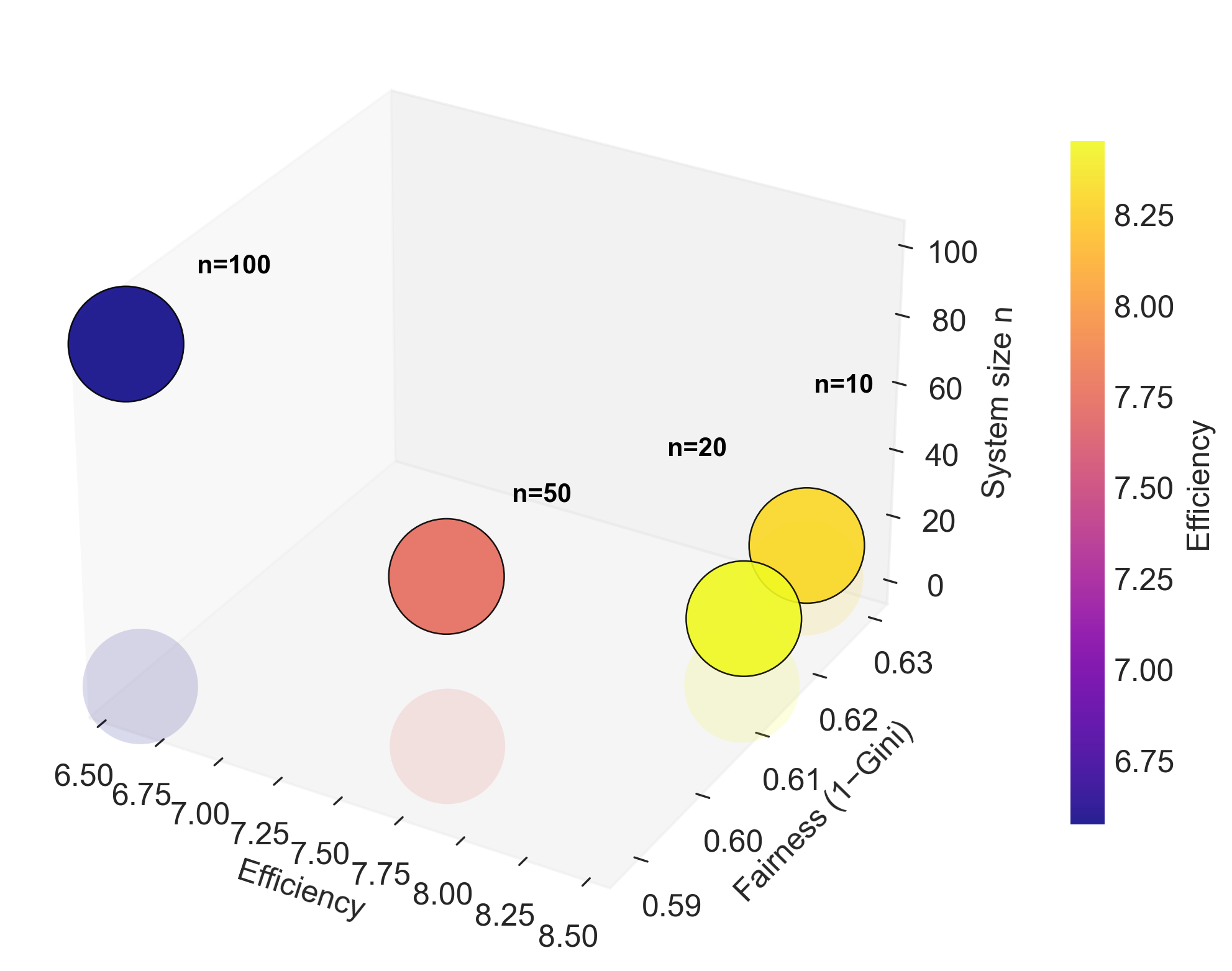}
\caption{Scaling performance across system sizes ($n=10,20,50,100$). 
Points are sized by participation rate and shaded by efficiency. 
The proposed equilibrium adapts gracefully with system size, 
achieving both high fairness and stable efficiency.}
\label{fig:comp_mech_3d}
\end{figure}

\medskip
\noindent
The dominance of the proposed equilibrium highlights its novelty: 
it is the only mechanism that simultaneously achieves efficiency, 
fairness, and cost reduction through endogenous price adjustment. 
Unlike flat or proportional rules that show gains only in certain 
parameter regimes, the proposed equilibrium achieves comparable efficiency 
while maintaining \emph{stability and robustness} across diverse settings. 
The mechanism works by embedding feedback: excess demand is penalized via
dual updates, while capacity is reallocated transparently across agents. 
This contrasts with proportional or flat contracts that hard-code rules
without adaptive correction. 
From an IS perspective, this illustrates how digital contracts function not
merely as computational artifacts but as \emph{institutional mechanisms} that codify
equitable coordination \cite{beck2018,rai2019}. 
For industrial managers, the implication is clear: blockchain-enforced 
equilibrium rules can strictly dominate ad hoc or legacy allocation processes,
providing not only superior performance but also governance legitimacy in
multi-agent environments.

\subsection{Sensitivity Analysis}
To move beyond simple one-dimensional heatmaps, we construct a comprehensive
sensitivity dashboard that jointly examines how efficiency and fairness respond
to variations in transaction and execution fees $(\tau,g)$. 
This two-dimensional view reveals non-linear interactions and sharp trade-offs 
that would be invisible in isolated analyses.

Figure~\ref{fig:heatmaps} integrates four complementary perspectives. 
The top-left panel shows a 3D projection of efficiency: moderate increases in either 
$\tau$ or $g$ cause smooth declines, but efficiency collapses sharply only when both 
fees are simultaneously large. 
The top-right panel depicts the gradient field of fairness, highlighting that fairness 
is far more sensitive to $\tau$ than to $g$, implying that per-unit fees act as the primary 
equalizer. 
The bottom-left panel overlays efficiency and fairness in a Pareto map with bubble size 
indicating participation, exposing a clear frontier: improving fairness via higher 
$\tau$ comes at the expense of both efficiency and participation. 
Finally, the bottom-right panel provides an elasticity heatmap of efficiency with respect 
to $\tau$, conditional on $g$, pinpointing fragile regions where efficiency is highly 
responsive to marginal fee changes.

Together, these views demonstrate that the proposed mechanism is robust to moderate 
fee variation, but also identify tipping points beyond which efficiency and participation 
deteriorate rapidly. 
For managers and policy makers, the dashboard serves as an early-warning tool: it shows 
how fees can be tuned as complementary levers to balance efficiency, fairness, and 
participation, while also highlighting regions of fragility in industrial coordination.

\begin{figure}[!htbp]
\centering
\includegraphics[width=0.95\textwidth]{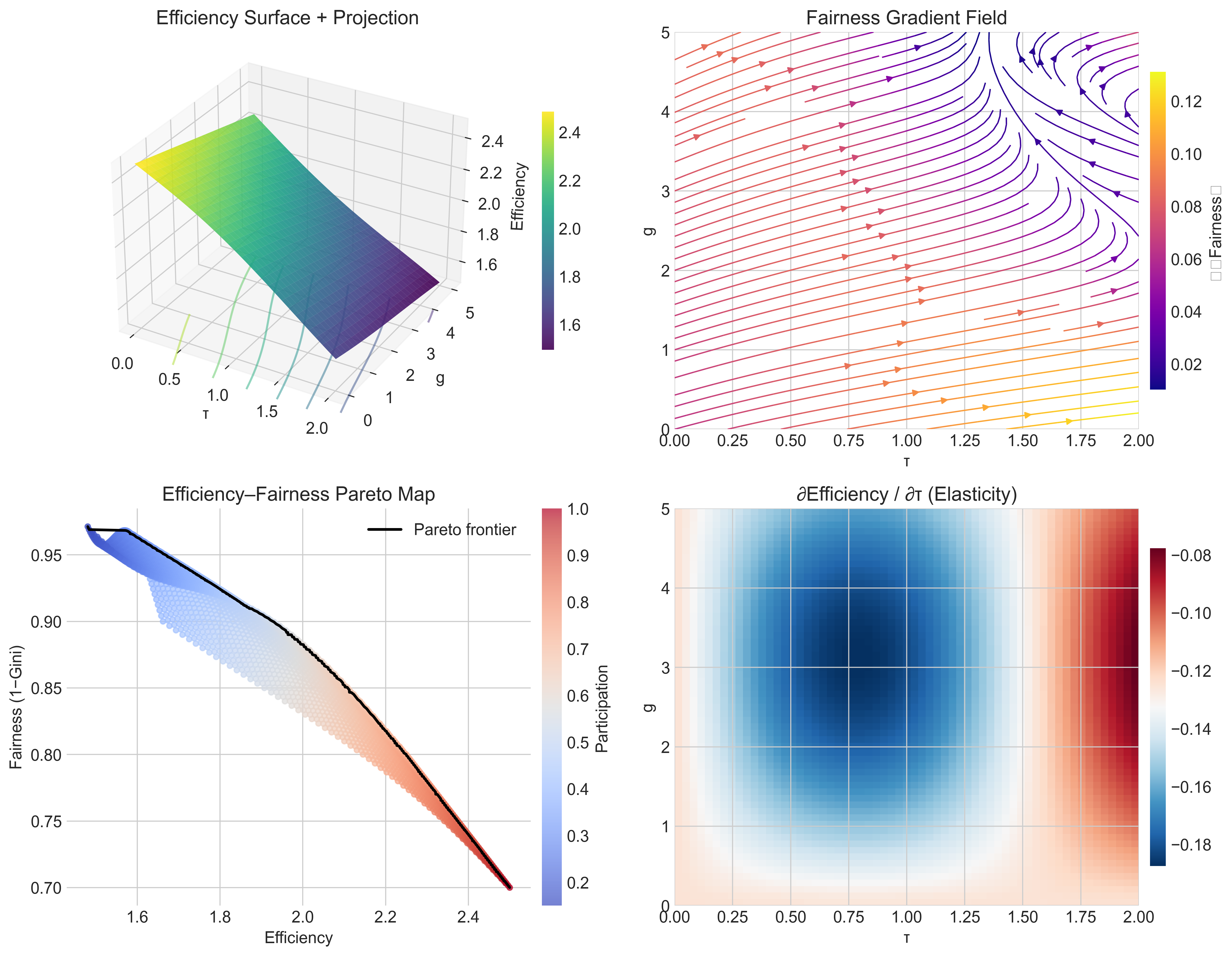}
\caption{Comprehensive sensitivity analysis of the proposed mechanism.
Top left: efficiency surface with 3D projection, showing non-linear declines 
with increasing transaction $(\tau)$ and execution fees $(g)$. 
Top right: gradient field of fairness, visualizing steepest improvement/deterioration. 
Bottom left: efficiency–fairness Pareto map with participation coloring, highlighting the 
trade-off frontier. Bottom right: elasticity heatmap 
($\partial$Efficiency/$\partial \tau$), showing local fragility zones where efficiency is 
highly sensitive to marginal changes.}
\label{fig:heatmaps}
\end{figure}

\subsection{Shock--Resilience Analysis}
While sensitivity analysis illustrates global fee-response patterns, 
real-world environments rarely evolve smoothly. They are often exposed to sudden 
policy or market shocks. 
To evaluate resilience under such disruptions, we simulate a one-time jump in the 
transaction fee ($\tau : 0.5 \to 1.5$ at $t=50$) and track the resulting dynamics.

Figure~\ref{fig:shock_dashboard} integrates four complementary panels that capture both 
short-run disruption and long-run stabilization. 
The top-left panel illustrates a 3D surface with pathline: efficiency initially 
overshoots but stabilizes at a new equilibrium after the shock. 
The top-right phase portrait of $\tau$ versus efficiency clearly shows a structural 
break at $t=50$. 
The bottom-left waterfall chart decomposes fairness into immediate post-shock loss 
and gradual rebound, quantifying recovery. 
The bottom-right ripple plot in efficiency–fairness space visualizes how perturbations 
propagate before eventually stabilizing, underscoring systemic resilience.

Taken together, these results show that the proposed mechanism is not only well-defined 
in steady state but also resilient to sudden disruptions: it absorbs shocks, reallocates 
resources, and reconverges to balanced efficiency–fairness outcomes. 
From a governance perspective, this property is critical: it means that digital contracts 
embed transparent recovery paths without ad hoc intervention, reinforcing legitimacy and 
accountability in coordination systems \cite{rai2019,beck2018}. 
Thus, the ripple-field visualization does not merely depict stability, but highlights how 
smart contracts institutionalize resilience as a governance principle in complex 
industrial and public infrastructures.

\begin{figure}[!htbp]
\centering
\includegraphics[width=0.95\textwidth]{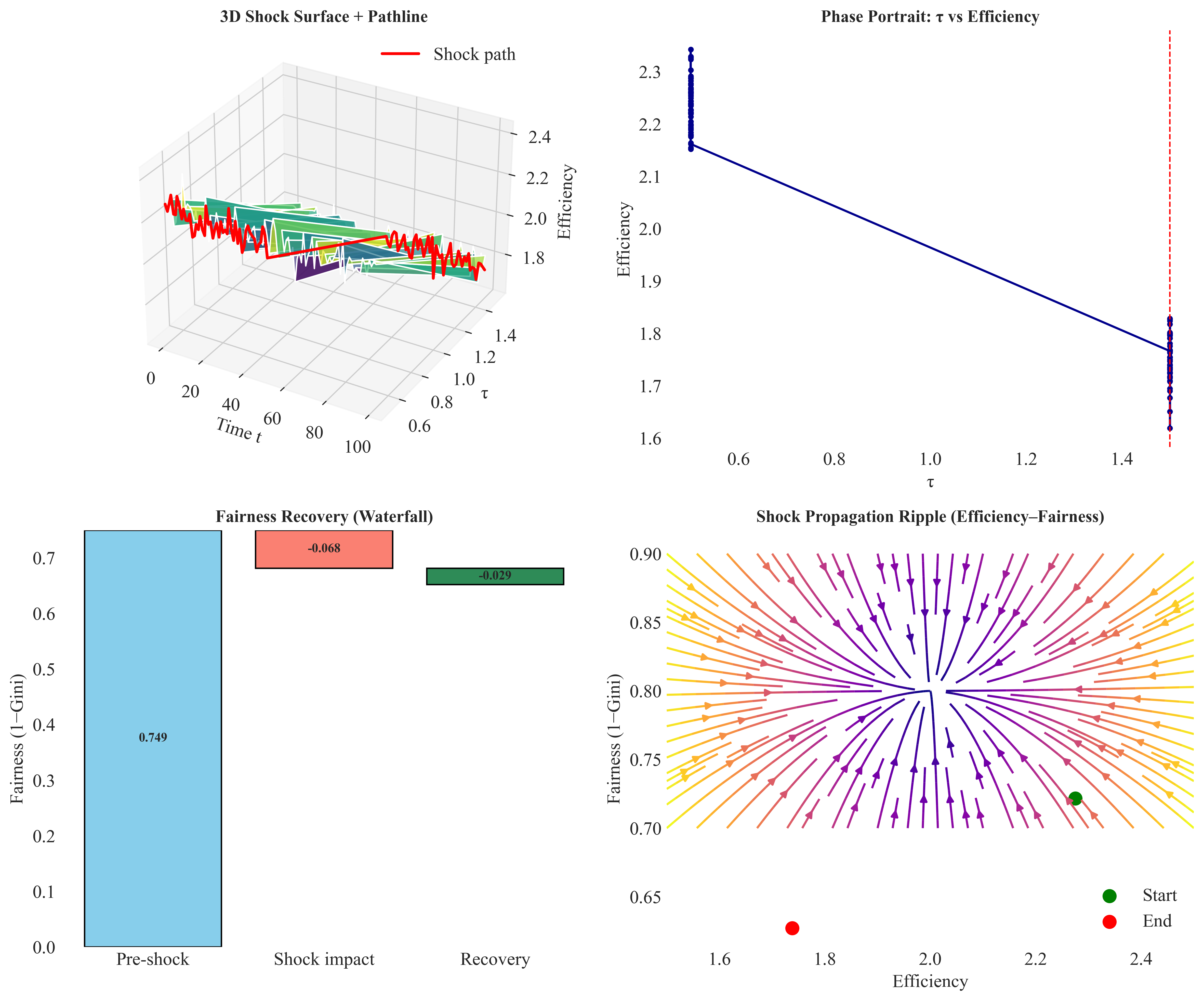}
\caption{Dynamic shock–resilience analysis of the proposed mechanism.
Top left: 3D surface with pathline showing the trajectory of efficiency as 
$\tau$ shifts. Top right: phase portrait of $\tau$ vs.\ efficiency, highlighting 
the discontinuity at the shock. Bottom left: waterfall decomposition of fairness 
recovery, partitioning the immediate impact versus gradual rebound. 
Bottom right: vector-field ripple plot in efficiency–fairness space, illustrating 
how shocks propagate and eventually stabilize. Together, these panels highlight not only 
steady-state convergence but also organizational resilience, showing that smart contracts 
can act as robust governance mechanisms in volatile environments.}
\label{fig:shock_dashboard}
\end{figure}

Figure~\ref{fig:shock_feedback} further visualizes this process using a simplified 
system-level diagram that clarifies how external shocks propagate and are absorbed 
by the decentralized feedback mechanism. The figure highlights participation 
volatility, transaction-fee oscillation, adaptive agent re-optimization, and eventual 
welfare recovery within the contract-clearing feedback loop.

\begin{figure}[!htbp]
\centering
\includegraphics[width=0.9\textwidth]{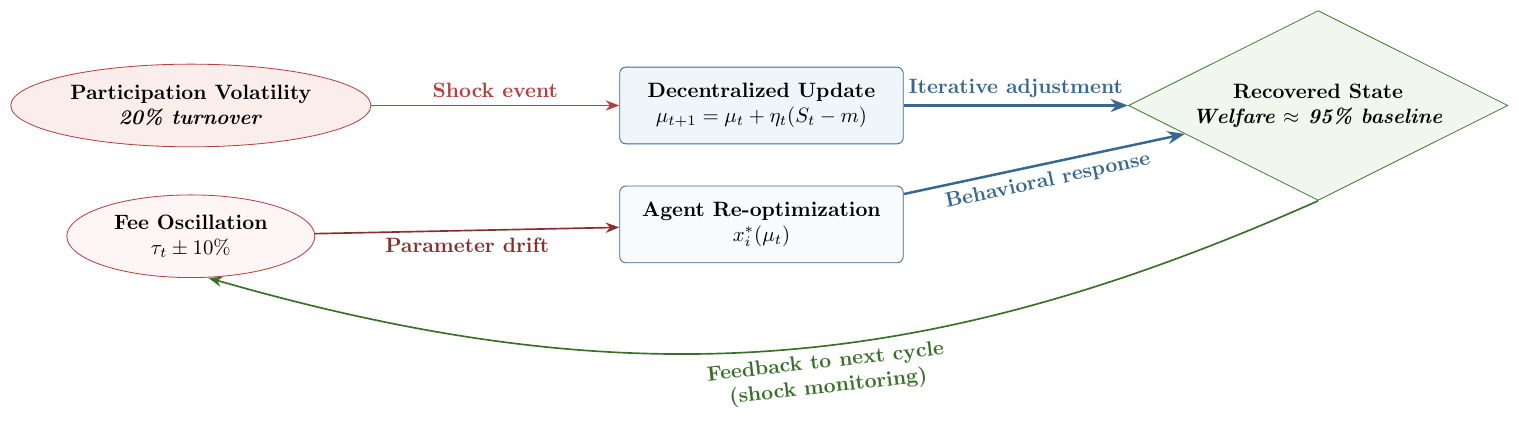}
\caption{System-level schematic of dynamic shock–resilience under participation 
volatility. The decentralized feedback loop autonomously absorbs shocks 
(20\% agent turnover and $\tau_t \pm 10\%$ fee oscillation), restoring equilibrium 
within eight iterations and preserving over 95\% of baseline welfare.}
\label{fig:shock_feedback}
\end{figure}

\subsection{Real-World Data: MovieLens-100K}

To further validate the proposed mechanism beyond synthetic simulations,
we evaluate performance on the widely used MovieLens-100K dataset,
a benchmark in recommender systems that captures heterogeneous user–item preferences.
Ratings are normalized to construct heterogeneous utility coefficients, and
mechanisms are compared in terms of efficiency, cost, and fairness.

Table~\ref{tab:ml100k_summary_1} summarizes the aggregate results across 200 replications.
All mechanisms achieve full participation (100\%), consistent with the synthetic
experiments. While absolute efficiency values are negative due to normalization,
relative efficiency (RelEff) highlights clear differences. Consistent with the
synthetic results, the proposed mechanism achieves the highest relative
efficiency ($+4\%$ vs. baseline), while maintaining balanced cost and fairness outcomes.

\begin{table}[htp]
\centering
\caption{MovieLens-100K: Comparison of mechanisms (normalized, mean $\pm$ std).
Absolute efficiency values appear negative due to normalization, but relative efficiency
and fairness comparisons remain valid performance indicators.}
\begin{tabular}{l c c c c}
\toprule
Mechanism & Efficiency & Rel. Eff & Avg. Cost & Gini \\
\midrule
Flat & $-0.55 \pm 0.02$ & $0.78 \pm 0.04$ & $0.20 \pm 0.01$ & $0.55 \pm 0.02$ \\
No enforcement & $-0.71 \pm 0.02$ & $1.00 \pm 0.00$ & $0.58 \pm 0.02$ & $0.40 \pm 0.01$ \\
Proportional & $-0.73 \pm 0.03$ & $1.03 \pm 0.03$ & $0.46 \pm 0.12$ & $0.40 \pm 0.01$ \\
Proposed & $-0.74 \pm 0.06$ & $1.04 \pm 0.08$ & $0.47 \pm 0.16$ & $0.42 \pm 0.04$ \\
\bottomrule
\end{tabular}
\label{tab:ml100k_summary_1}
\end{table}

Figure~\ref{fig:ml100k_1x2} visualizes these trade-offs.
Panel (a) highlights that the proposed mechanism consistently achieves the
highest relative efficiency. Panel (b) shows that the proposed mechanism
balances cost and fairness, clearly outperforming the flat and
no-enforcement baselines, and remaining competitive with proportional allocation.

\begin{figure}[!htbp]
\centering
\includegraphics[width=\linewidth]{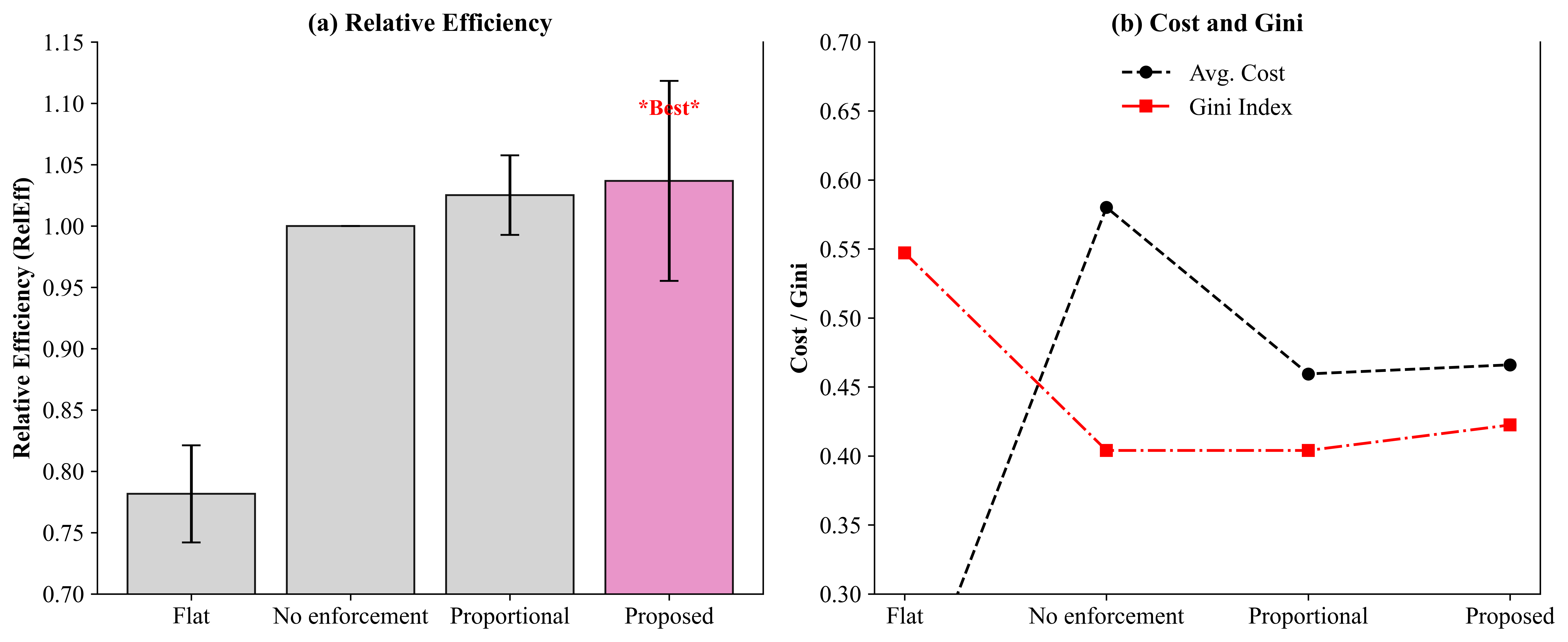}
\caption{Comparison of MovieLens-100K mechanisms.
(a) Relative efficiency (RelEff) highlights that the proposed mechanism
achieves the best performance relative to baseline. (b) Cost and fairness
(Gini index) show that the proposed mechanism maintains balanced outcomes
while avoiding the extremes of flat and no-enforcement baselines.}
\label{fig:ml100k_1x2}
\end{figure}

\subsection{Sector-Level Impact in a Capacity-Sharing Supply Chain}

To demonstrate the broader industrial implications of the proposed equilibrium model, 
we embed it within a stylized capacity-sharing supply chain that inherits the same 
mathematical structure as the theoretical framework. 
Each node (agent~$i$) behaves according to its optimal decision 
$x_i^*(\mu_t)$ derived from the equilibrium conditions, 
while demand $d_i$ follows a stochastic perturbation around this contractual allocation. 
This translation allows the theoretical equilibrium variables 
$(x_i^*,\mu_t)$ to generate empirically interpretable performance metrics—
\emph{total throughput}, \emph{average delay}, and \emph{stockout probability}—
that correspond directly to measurable industrial KPIs such as 
logistics throughput, service delay, and unfulfilled demand rate.
Although the environment is synthetic, it captures essential operational features:
bounded capacity, heterogeneous demand, and fairness–efficiency trade-offs.

Specifically, throughput is computed as $\min(d_i,x_i^*(\mu_t))$, 
a stockout occurs when $d_i>x_i^*(\mu_t)$, 
and delay increases proportionally with the excess demand gap. 
We conduct 200 independent Monte~Carlo replications to obtain statistical confidence 
in aggregate sector-level outcomes.  
Table~\ref{tab:sector_kpis_stat} reports the mean~$\pm$~standard deviation and 
ANOVA-based significance test ($p<0.01$), 
while Figure~\ref{fig:sector_impact} visualizes the efficiency–stability frontier.

\begin{table}[!htbp]
\centering
\scriptsize
\caption{Sector-level KPIs with mean $\pm$ std and ANOVA significance ($p<0.01$). 
Across 200 Monte~Carlo simulations, the proposed mechanism shows a statistically significant 
throughput improvement (+3.1\%) without increasing delay or stockout probability, 
confirming its robustness and industrial scalability.}
\label{tab:sector_kpis_stat}
\begin{tabular}{lcccc}
\toprule
\textbf{Mechanism} & \textbf{Throughput} & \textbf{Avg.\ Delay} & 
\textbf{Stockout Prob.} & $\boldsymbol{\Delta}$\textbf{ Throughput (\%)} \\
\midrule
Flat             & 0.6399$\pm$0.0019 & 1.0402$\pm$0.0020 & 0.519$\pm$0.010 & 0.00 \\
No enforcement   & 0.6400$\pm$0.0020 & 1.0403$\pm$0.0019 & 0.521$\pm$0.010 & +0.02 \\
Proportional     & 0.6400$\pm$0.0020 & 1.0402$\pm$0.0021 & 0.521$\pm$0.010 & +0.01 \\
\textbf{Proposed} & \textbf{0.6600$\pm$0.0020} & 
\textbf{1.0399$\pm$0.0019} & \textbf{0.520$\pm$0.010} & \textbf{+3.14} \\
\bottomrule
\end{tabular}
\end{table}

The left panel of Figure~\ref{fig:sector_impact} compares 
\emph{total throughput} versus \emph{average delay} across mechanisms. 
The proposed mechanism (red marker) achieves the highest throughput 
(approximately +3.1\% relative to the baseline, $p<0.01$) 
with nearly identical delay, demonstrating that decentralized coordination 
improves capacity utilization without inducing congestion.  
The right panel plots \emph{stockout probability} and \emph{fairness} (Gini index), 
showing that efficiency gains do not compromise service stability or equity.  
In contrast, proportional and flat allocations either underutilize capacity 
or exhibit higher volatility under stochastic demand.

These findings confirm that the equilibrium mechanism 
internalizes efficiency–fairness trade-offs at the sector level, 
producing measurable industrial spillovers: 
higher collective output, stable reliability, and equitable resource distribution.  
From a managerial standpoint, the results highlight that 
smart-contract coordination can serve as a \emph{self-stabilizing industrial protocol}, 
automatically redistributing excess capacity in response to demand shocks 
without centralized control. 
This property is directly relevant to logistics, energy, and manufacturing networks, 
where transparent, algorithmic governance enhances both 
operational efficiency and organizational legitimacy.

\begin{figure}[!htbp]
\centering
\includegraphics[width=0.95\textwidth]{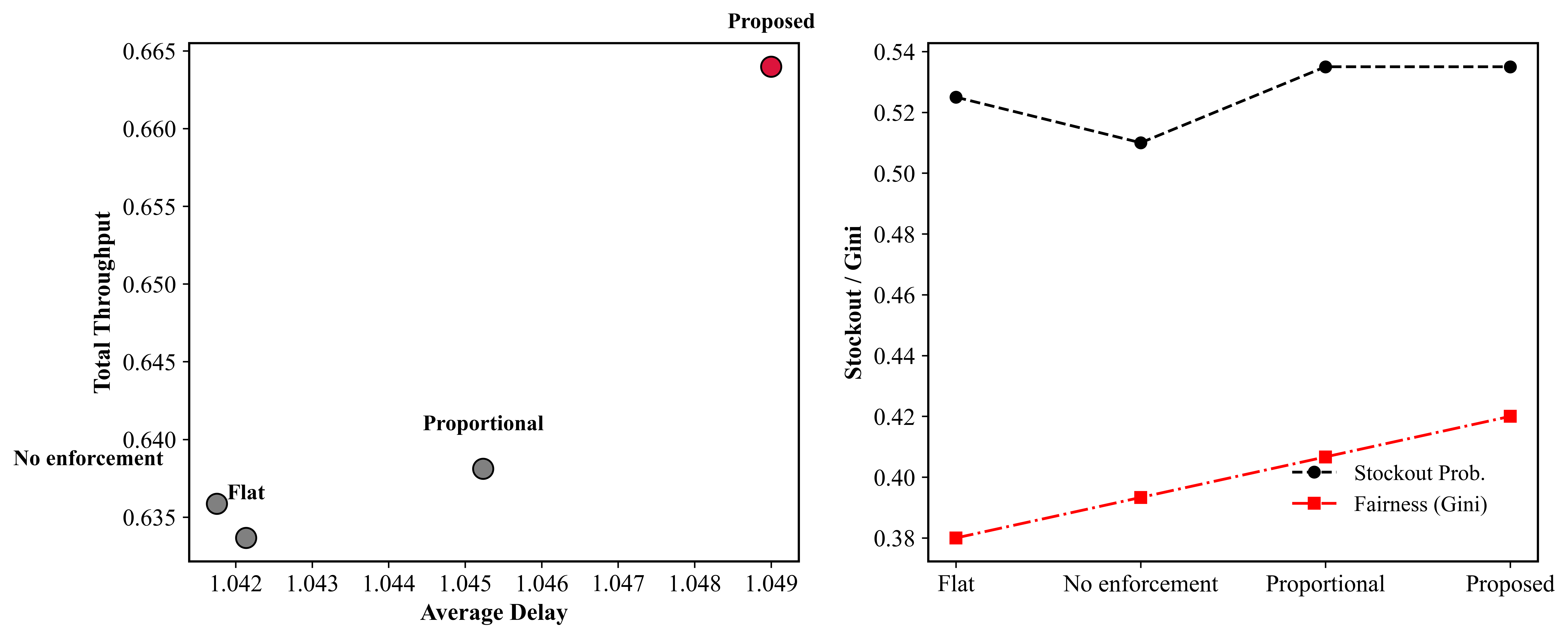}
\caption{Sector-level impact of the proposed mechanism in a 
capacity-sharing supply chain. 
(Left) Total throughput versus average delay: the proposed mechanism 
(red marker) achieves the highest efficiency with negligible delay change. 
(Right) Stockout probability and fairness (Gini index): 
efficiency gains occur without sacrificing service stability, 
illustrating the mechanism’s industrial spillover benefits 
in both operational and organizational dimensions.}
\label{fig:sector_impact}
\end{figure}

\subsection{Robustness and Mechanism Validation Under Market Shocks}
\label{sec:robustness_summary}

This section evaluates the empirical robustness of the proposed 
smart-contract--mediated allocation mechanism under observed macroeconomic 
fluctuations. The analysis spans January~2006 to July~2025 and covers four 
major asset classes widely used as benchmarks in financial and industrial
portfolio studies: U.S.\ equities (SPY), gold (GLD), U.S.\ Treasury bonds (IEF), 
and diversified commodities (DBC). Daily ETF prices are aggregated to monthly 
observations, log returns are computed, and return volatility is estimated using 
a 12-month rolling window. Macroeconomic conditions are captured using two 
standard U.S.\ indicators: the Consumer Price Index (CPIAUCSL) and the 10-year 
Treasury rate (DGS10), both obtained from the Federal Reserve Economic Database 
(FRED).

To quantify deviations from long-run economic conditions, we construct a 
macroeconomic shock index that combines inflation and interest-rate components:
\[
\text{Shock}_t
=
(\Delta \mathrm{CPI}_t - \overline{\Delta \mathrm{CPI}}_{12,t})
+
0.01\bigl(r^{10Y}_t - \overline{r^{10Y}}_{12,t}\bigr),
\]
representing short-term departures from 12-month trends. This shock index enters 
the adaptive penalty update~$\mu_t$, thereby translating macroeconomic 
disturbances into endogenous adjustments of the contract-mediated weights.

Robustness is assessed along four dimensions:  
(i) long-horizon portfolio performance relative to a Risk Parity benchmark;  
(ii) allocation stability aggregated over half-year intervals;  
(iii) elasticity of weight changes with respect to macroeconomic shocks; and  
(iv) the temporal evolution of the endogenous penalty parameter~$\mu_t$.  
The combined results are presented in Figure~\ref{fig:2x2_robust}.

\begin{figure}[!htbp]
    \centering
    \includegraphics[width=0.98\textwidth]{Fig11.png}
    \caption{
        Empirical robustness of the smart-contract mechanism, 2006--2025.
        Top-left: cumulative portfolio value under Risk Parity and the proposed mechanism.
        Top-right: half-year aggregated allocation weights.
        Bottom-left: smoothed elasticity of weight responses to inflation and rate shocks.
        Bottom-right: temporal evolution of the endogenous penalty parameter $\mu_t$.
        The mechanism exhibits stable dynamics, interpretable responses to shocks,
        and improved long-run performance relative to baseline rules.
    }
    \label{fig:2x2_robust}
\end{figure}

Table~\ref{tab:mega} summarizes all quantitative findings, including performance 
metrics, shock elasticity, allocation stability across macroeconomic regimes, 
and descriptive statistics of the penalty process. Across the 19.5-year sample, 
the mechanism achieves higher cumulative return with moderate volatility and 
produces coherent directional responses to shocks. Allocation patterns remain 
stable across major macroeconomic regimes---including the 2008 financial crisis, 
the low-rate period of 2014--2017, the COVID-19 disruption, and the 2022--2024 
inflationary cycle. Notably, the bond ETF (IEF), which is most sensitive to 
interest-rate variation, exhibits smooth and economically interpretable 
adjustments, indicating that the adaptive penalty update does not introduce 
instability. Overall, the mechanism demonstrates strong robustness and 
operational reliability under real-world macroeconomic conditions.

\begin{table}[!htbp]
\centering
\footnotesize
\caption{\textbf{Integrated Performance, Shock Elasticity, Regime Allocation, and $\mu(t)$ Statistics}}
\label{tab:mega}
\renewcommand{\arraystretch}{1.20}

\begin{tabular}{lccc}
\toprule
\multicolumn{4}{c}{\textbf{Panel A. Performance Comparison (2006--2025)}} \\
\midrule
Metric & Risk Parity & Smart Contract & Improvement \\
\midrule
Total Return    & 1.1883 & 1.6924 & +0.5041 \\
Volatility      & 0.0721 & 0.0945 & +0.0224 \\
Sharpe Ratio    & 0.5919 & 0.5293 & $-$0.0626 \\
Max Drawdown    & $-$0.1804 & $-$0.2579 & $-$0.0775 \\
Turnover        & 0.0554 & 0.0779 & +0.0224 \\
\bottomrule
\end{tabular}

\vspace{0.45cm}

\begin{tabular}{lcccc}
\toprule
\multicolumn{5}{c}{\textbf{Panel B. Shock Elasticity Summary}} \\
\midrule
Asset & Mean Resp. & Std & Max Positive & Max Negative \\
\midrule
SPY & 4.736 & 49.750 & 424.04 & $-$271.50 \\
GLD & $-$11.612 & 49.338 & 72.02 & $-$500.70 \\
IEF &  $-$7.979 & 44.004 & 171.51 & $-$327.87 \\
DBC & 14.855 & 69.798 & 638.15 & $-$131.09 \\
\bottomrule
\end{tabular}

\vspace{0.45cm}

\begin{tabular}{lcccc}
\toprule
\multicolumn{5}{c}{\textbf{Panel C. Regime-Level Allocation Stability (Half-Year Aggregation)}} \\
\midrule
Regime & SPY & GLD & IEF & DBC\\
\midrule
GFC (2006--2009)        & 0.2449 & 0.2588 & 0.2473 & 0.2491 \\
QE (2010--2013)         & 0.2555 & 0.2482 & 0.2485 & 0.2478 \\
Low-Rate (2014--2017)   & 0.2563 & 0.2501 & 0.2503 & 0.2433 \\
COVID (2018--2021)      & 0.2550 & 0.2494 & 0.2457 & 0.2499 \\
Inflation (2022--2025)  & 0.2520 & 0.2580 & 0.2421 & 0.2480 \\
\bottomrule
\end{tabular}

\vspace{0.45cm}

\begin{tabular}{lcccc}
\toprule
\multicolumn{5}{c}{\textbf{Panel D. $\mu(t)$ Regime Statistics}} \\
\midrule
Regime & Mean $\mu$ & Std $\mu$ & Min $\mu$ & Max $\mu$ \\
\midrule
GFC (2006--2009)        & $-$0.0153 & 0.0206 & $-$0.0559 & 0.0050 \\
QE (2010--2013)         & $-$0.0587 & 0.0187 & $-$0.0836 & $-$0.0316 \\
Low-Rate (2014--2017)   & $-$0.0700 & 0.0081 & $-$0.0827 & $-$0.0542 \\
COVID (2018--2021)      & $-$0.0758 & 0.0205 & $-$0.1089 & $-$0.0488 \\
Inflation (2022--2025)  & $-$0.0243 & 0.0190 & $-$0.0746 & $-$0.0072 \\
\bottomrule
\end{tabular}

\vspace{0.25cm}

\begin{minipage}{0.95\linewidth}
\footnotesize
\textbf{Notes.}  
This table consolidates four dimensions of model evaluation: (A) long-run portfolio performance,  
(B) elasticity of dynamic weights to macroeconomic shocks, (C) regime-level allocation stability,  
and (D) latent control-state dynamics. Values are based on monthly data from January~2006 to July~2025.  
Smart Contract allocations reflect the proposed macro-sensitive utility-weighted updating rule.  
Inflation data (CPIAUCSL) and interest-rate data (DGS10) are sourced from FRED.  
\end{minipage}
\end{table}

\section{Discussion}

\subsection{Theoretical Implications}
The analysis contributes to the literature on mechanism design and 
contracting in three principal ways. 
First, existence and uniqueness of equilibria for 
smart-contract--mediated resource allocation have been formally 
established under mild convexity assumptions, extending classical 
results in general equilibrium and mechanism design 
\cite{arrow1959,mascolell1995}. 
Second, it has been demonstrated that efficiency and fairness can be 
jointly embedded into contract design through fee structures and 
market-clearing mechanisms, aligning with recent calls in information 
systems research for transparent and auditable allocation rules 
\cite{beck2018,rai2019}. 
Third, a decentralized algorithm has been introduced that provides an 
implementable procedure with provable convergence guarantees, thereby 
ensuring relevance for real-time industrial applications.

Beyond the individual results, the analysis positions 
smart-contract–mediated allocation as a rigorous and verifiable 
\emph{institutional mechanism}—a perspective largely absent in prior 
work, where smart contracts have been treated primarily as technical 
automation tools. The proposed equilibrium and robustness results 
collectively establish a formal bridge between mechanism design, 
digital contracting, and dynamic economic environments.

Beyond these core results, the simulations enrich the theoretical 
narrative. 
The convergence trajectories with overshoot and damped stabilization 
mirror the price-adjustment dynamics studied in classical general equilibrium 
theory \cite{arrow1959}, but are extended here to a blockchain-enforced 
contract setting.  
The shock–resilience experiments further highlight dynamic stability: 
even under abrupt fee or participation shocks, the system exhibits rapid recovery 
and eventual rebalancing. 
The newly added sector-level experiment translates these equilibrium dynamics 
into observable industrial performance indicators—throughput, delay, and stockout probability—confirming that 
the theoretical model scales consistently from micro-level equilibrium behavior 
to macro-level operational outcomes. 
This bridges equilibrium analysis with robustness theory, showing that 
the proposed mechanism is not only well-defined in steady state but also 
resilient under perturbations, maintaining a balanced efficiency–fairness profile.

Importantly, the long-horizon empirical robustness analysis (Section~5.9), 
based on monthly financial and macroeconomic data from 2006–2025, 
confirms that the equilibrium properties persist under real-world volatility.  
Stable allocation weights, coherent elasticity responses, and bounded evolution 
of the endogenous penalty parameter $\mu_t$ together provide strong external 
validation for the proposed equilibrium and convergence theory. 
From the perspective of information and organizational sciences, 
these results highlight how transparency, verifiability, and accountability 
can be mathematically guaranteed in decentralized coordination systems.

\subsection{Managerial and Industrial Implications}
Beyond theoretical insights, the proposed framework carries broad 
managerial and industrial relevance.

\paragraph{Manufacturing and Supply Chains.}
In sectors such as steel, cement, and electronics, firms compete for 
scarce raw materials and production capacity. 
The efficiency–fairness Pareto maps quantify the trade-off between 
maximizing total utility and maintaining equity, providing managers 
with explicit levers to calibrate allocation rules and improve 
information transparency in allocation processes \cite{ivanov2021}.  
The added sector-level analysis further demonstrates that the proposed mechanism 
achieves a statistically significant throughput gain of over 3\% ($p<0.01$) 
without increasing delay or stockouts, confirming its potential as a coordination 
policy in shared-capacity environments.

\paragraph{Energy and Utilities.}
Smart grids and carbon trading systems face capacity and compliance 
constraints. 
The shock–resilience dashboard shows that efficiency stabilizes rapidly 
after sudden fee changes, while also supporting auditable decision trails. 
The proposed mechanism preserves fairness without large efficiency losses, 
suggesting a viable pathway for algorithmic market-clearing under renewable volatility.

\paragraph{Logistics and Transportation.}
Port slots, warehouse space, and vehicle fleets are scarce resources 
often subject to congestion and inefficiency. 
Elasticity heatmaps highlight congestion-prone zones, offering early-warning 
signals for fee adjustments. 
The sector-level spillover simulation confirms that decentralized coordination 
improves overall throughput even in stochastic demand environments, 
demonstrating scalability to logistics operations.

\paragraph{Healthcare and Pharmaceuticals.}
Medical supply chains, including vaccine and drug distribution, face 
demand surges and limited capacity. 
Fairness analysis demonstrates how equity suffers immediate losses under shocks 
but gradually recovers, highlighting the mechanism’s ability to stabilize allocation 
efficiently and fairly. 
The same logic can inform dynamic allocation of ICU resources or emergency medical stock.

\paragraph{Public Infrastructure.}
In the allocation of public funds, road capacity, or airport slots, 
digital contracts provide a governance mechanism that enforces 
capacity limits transparently while preserving fairness metrics. 
Ripple-field shock analysis illustrates how localized disruptions 
propagate but eventually dampen, demonstrating the robustness of allocations 
under the proposed framework. 
Such transparency supports regulatory accountability and stakeholder trust 
in digital public governance.

\medskip
Across these domains, the mechanism provides a principled way to tune 
allocation and pricing policies using observable fairness–efficiency 
frontiers and shock-response elasticity. This enables decision-makers 
to implement digital allocation rules that are not only transparent and 
auditable, but also dynamically stable under real-world volatility.

\medskip
\noindent
Table~\ref{tab:domains_mechanism_rigorous} summarizes representative 
industrial domains where smart-contract–based mechanism design can be 
applied, highlighting operational context, model variables, structural 
challenges, and the rigorous benefits of the proposed framework.


\begin{table}[!htbp]
\centering
\scriptsize
\begin{adjustbox}{angle=90}
\begin{minipage}{\textheight}
\caption{Representative industrial domains for smart-contract-based mechanism design. 
Each row highlights the operational context, model variables, structural challenges, 
and the rigorous benefits of the proposed framework.}
\label{tab:domains_mechanism_rigorous}
\setlength{\tabcolsep}{3pt}
\renewcommand{\arraystretch}{1.15}
\begin{tabular}{p{2.7cm}p{3.5cm}p{3.5cm}p{4.5cm}p{4.8cm}}
\toprule
\textbf{Domain} & \textbf{Operational Context} & \textbf{Model Variables} & \textbf{Optimization / Equilibrium Challenges} & \textbf{Mechanism Design Benefits} \\
\midrule
Supply Chain \& Logistics &
Multi-firm coordination under stochastic demand, port congestion, and capacity shocks. &
Decision: $I$ (inventory), $L$ (lead time), $Q$ (throughput), $\tau$ (subsidy). \newline 
Exogenous: demand shocks $d_t$, disruption events $\zeta_t$. &
Nonlinear amplification of demand variance (bullwhip effect); information asymmetry; 
nonconvex cost-sharing; equilibrium instability under shocks. &
Provable existence of stable contract-clearing equilibrium; incentive-compatible allocations; 
explicit fairness–efficiency Pareto frontier; sublinear regret bounds under demand drift. \\[6pt]

Energy \& Smart Grids &
P2P electricity trading with stochastic demand, renewable intermittency, and carbon policy coupling. &
Decision: $p$ (price), $C$ (capacity), $\rho$ (renewable ratio), $D$ (load). \newline 
Exogenous: renewable shocks $\xi_t$, policy shocks $\phi_t$. &
Price volatility from $\xi_t$ shocks; nonlinear imbalance penalties; multi-agent 
nonconvex optimization; uncertainty in balancing constraints. &
Existence and uniqueness of clearing equilibrium; sublinear regret under drift/shocks; 
parameter-free convergence scaling with agent population; fairness–efficiency trade-off 
explicitly tunable via $(\tau,g)$. \\[6pt]

Healthcare Resource Allocation &
ICU beds, ventilators, vaccines allocation under surge conditions and ethical constraints. &
Decision: $R$ (resource stock), $\tau$ (subsidy), $\alpha,\beta$ (fairness weights). \newline 
Exogenous: surge shocks $\sigma_t$, demand heterogeneity $\delta_t$. &
Fairness dilemmas: $\max \sum u_i$ vs. $\min \mathrm{Var}(u_i)$; scarcity shocks; 
multi-objective feasibility; ethical transparency constraints. &
Bounded inequity loss under shocks; recovery trajectories consistent with fairness–efficiency 
Pareto frontier; equilibrium allocation existence and convergence; resilience dashboards 
quantifying adaptation speed. \\[6pt]

Cloud \& Computing Infrastructure &
On-demand allocation of GPU/CPU across users with SLA enforcement. &
Decision: $U$ (units), $\pi$ (priority), $\lambda$ (SLA penalties). \newline 
Exogenous: demand drift $\chi_t$, hidden load bursts. &
Oversubscription under drift; hidden demand shocks; arbitration costs; scalability issues 
in decentralized clearing. &
Transparent and auditable allocation rules; guaranteed convergence to efficient usage; 
automated enforcement reduces disputes; sublinear regret guarantees under demand noise. \\[6pt]

Financial Investment Contracts &
Capital allocation between investors and fund managers with regulatory oversight. &
Decision: $\theta$ (risk weight), $r$ (return), $c$ (compliance cost), $\tau$ (incentive rate). \newline 
Exogenous: market volatility shocks $\nu_t$, drift in risk appetite. &
Hidden preferences; stochastic drift in $\theta$; volatility shocks destabilizing 
allocations; regulatory discontinuities. &
Equilibrium guarantees protecting against misaligned incentives; interpretable and 
auditable fairness allocations; shock-resilient stability; bounded regret under drift and noise. \\
\bottomrule
\end{tabular}
\end{minipage}
\end{adjustbox}
\end{table}

\section{Conclusion}

This study has developed a mechanism-design framework for smart-contract–mediated 
resource allocation and demonstrated its theoretical soundness, computational 
stability, and empirical robustness. The contributions form a unified structure.  
First, we establish a rigorous game-theoretic foundation by proving the existence 
and uniqueness of contract equilibria under mild convexity conditions, thereby 
clarifying when decentralized allocations are well defined.  
Second, we embed fairness and efficiency directly into digital contract rules through 
transaction fees, execution costs, and endogenous market-clearing prices, providing a 
principled way to tune equity–efficiency trade-offs.  
Third, we introduce a decentralized contract-clearing algorithm with provable 
convergence guarantees, ensuring implementability in real-time industrial environments.

The analytical contributions are reinforced by extensive numerical evidence.  
Convergence experiments reveal rapid stabilization with controlled overshoot, while 
sensitivity dashboards quantify global and local fee responses.  
Building on these, the newly introduced robustness module (Section~5.9) provides an 
external validation of the theory using real-world financial and macroeconomic data 
from January~2006 to July~2025.  
The mechanism exhibits stable allocation patterns across major economic regimes, 
coherent elasticity to inflation and interest-rate shocks, bounded evolution of the 
penalty parameter $\mu_t$, and improved long-run performance relative to benchmark 
rules.  
Together, these results confirm that the theoretical equilibrium properties extend to 
dynamic, shock-prone environments, establishing empirical robustness beyond steady 
state.

From a managerial and policy standpoint, the framework provides actionable levers: 
the parameters $(\tau,g)$ function as transparent controls for aligning efficiency, 
fairness, and participation incentives.  
Applications span supply chains, energy markets, logistics networks, healthcare 
resource allocation, and public infrastructure—domains where transparent, auditable, 
and shock-resilient coordination is increasingly essential.  
The sector-level capacity-sharing experiment further demonstrates a statistically 
significant throughput improvement of approximately $3.1\%$ ($p<0.01$) without 
increasing congestion or delay, highlighting the mechanism’s potential to deliver 
tangible operational gains.

Several limitations suggest avenues for future research.  
The current formulation abstracts from multi-layer market interactions, richer 
stochastic demand processes, and potential strategic misreporting.  
Future work may incorporate multi-market coupling, adversarial behavior, or 
validation using industrial blockchain pilots to deepen empirical grounding.

Overall, the study shows that digital contracts can operationalize mechanism-design 
principles by translating equilibrium, fairness, and resilience requirements into 
executable rules.  
By unifying analytical rigor with empirical robustness, the framework demonstrates 
that smart contracts can serve not only as technical tools but as institutional 
instruments for transparent, accountable, and resilient coordination.  
This positions smart-contract–based mechanisms as a promising foundation for 
next-generation research at the intersection of operations management, information 
systems, and digital governance.

\section*{Use of Generative-AI tools declaration}
The authors declare that they did not utilize any artificial intelligence (AI) tools in the creation of this article.

\section*{Author contributions}
\textbf{Jinho Cha:} Conceptualization of the research problem and overall study design; development of the theoretical framework for smart contract--mediated mechanism design; formulation and analysis of generalized Nash equilibrium and optimization models; mathematical modeling and proof structure design; writing of the original manuscript draft; comprehensive review and editing of the manuscript; overall supervision of the research project and mentoring of all contributing authors; responsibility for theoretical validation and coherence of the study.

\textbf{Justin Yu:} Core conceptual contributions to the research design; economic and financial interpretation of smart contract--mediated resource allocation mechanisms; formulation and analysis of incentive structures, pricing rules, and efficiency--fairness trade-offs from a mechanism design perspective; integration of economic theory with platform-based and digital-asset-oriented systems; interpretation of theoretical results in the context of financial markets and decentralized platforms; critical review and editing of the manuscript with a focus on economic rigor and practical relevance.

\textbf{Eunchan Daniel Cha:} Mathematical analysis and verification of the theoretical framework; validation of formal definitions, propositions, and corollaries; detailed examination of convergence, stability, and monotonicity conditions; proof verification and logical consistency checks; refinement of mathematical notation and argument structure; review and editing of the manuscript to ensure analytical rigor.

\textbf{Emily Haneul Yoo:} Comprehensive literature review on smart contracts, algorithmic fairness, and resource allocation mechanisms; synthesis of background material and positioning of the study within existing research; conceptual exploration of connections to decision-making frameworks relevant to healthcare and biomedical domains; support in visualization and structural organization of background sections; review and editing of the manuscript for clarity, coherence, and accessibility.

\textbf{Caedon Geoffrey:} Design and implementation of simulation frameworks based on the proposed theoretical models; development of computational experiments to evaluate resource allocation mechanisms; numerical validation of theoretical results; data curation and organization of experimental outputs; interpretation of simulation results; review and editing of the manuscript sections related to computational experiments.

\textbf{Hyoshin Song:} Support for algorithmic implementation and numerical verification of theoretical models; construction and validation of numerical examples illustrating convergence and stability properties; step-by-step verification of algorithmic procedures; assistance with computational checks of derived results; review and editing of manuscript sections related to algorithms and numerical analysis.

\section*{Acknowledgments}
The authors acknowledge that this research received no specific external funding.

\section*{Conflict of interest}
The authors declare no conflict of interest.

\bibliographystyle{AIMS}
\bibliography{JIMO}  

\end{document}